\documentclass[letterpaper,11pt]{article}
\usepackage{fullpage}
\usepackage{comment}
\usepackage[whole,autotilde]{bxcjkjatype}
\usepackage{amsthm,amsmath,amssymb}
\usepackage{algorithm,algpseudocode}
\usepackage[unicode]{hyperref}
\usepackage{booktabs}
\usepackage{authblk}
\usepackage{cite}
\usepackage{geometry}
\usepackage{enumerate}
  
\theoremstyle{definition}
\newtheorem{theorem}{Theorem}[section]
\newtheorem{Claim}       {Claim}[section]
\newtheorem{lemma}      {Lemma}[section]

\newtheorem{corollary}   {Corollary}[section]

\newtheorem{definition}  {Definition}[section]
\newtheorem{example}     {Example}[section]
\newtheorem{remark}      {Remark}[section]
\newtheorem{assumption}  {Assumption}[section]

\numberwithin{equation}{section}
\numberwithin{figure}{section}
\numberwithin{table}{section}

\newcommand{\floor}[1]{\left\lfloor #1 \right\rfloor}

\newcommand{\E}{\mathbb{E}}
\newcommand{\supp}{\mathrm{supp}}
\newcommand{\OMIT}[1]{}

\makeatletter
\@addtoreset{equation}{section}
\makeatother

\usepackage{color}
\usepackage{CJKutf8}
\usepackage{ifthen}
\newcommand{\COMM}[2]{{
\begin{CJK}{UTF8}{ipxm}
\ifthenelse{\equal{#1}{YY}}{\color{blue}}{
\ifthenelse{\equal{#1}{TM}}{\color{red}}{
\ifthenelse{\equal{#1}{AA}}{\color{cyan}}{
\ifthenelse{\equal{#1}{BB}}{\color{magenta}}}}}
[#1: #2]
\end{CJK}
}}

\title{Stochastic Packing Integer Programs with Few Queries\footnote{This is the final draft of the paper accepted for publication in Mathematical Programming (Series A). A preliminary version of this paper appeared in SODA 2018.}}
\author[1]{Takanori Maehara\footnote{Email: takanori.maehara@riken.jp}}
\affil[1]{RIKEN Center for Advanced Intelligence Project}
\author[2,1]{Yutaro Yamaguchi\footnote{Email: yutaro\_yamaguchi@ist.osaka-u.ac.jp}}
\affil[2]{Osaka University}
\date{\empty}

\begin{document}
\maketitle

\thispagestyle{empty}
\begin{abstract}
We consider a stochastic variant of the packing-type integer linear programming problem,
which contains random variables in the objective vector.
We are allowed to reveal each entry of the objective vector by conducting a query,
and the task is to find a good solution by conducting a small number of queries.
We propose a general framework of adaptive and non-adaptive algorithms for this problem, and provide a unified methodology for analyzing the performance of those algorithms.
We also demonstrate our framework
by applying it to a variety of stochastic combinatorial optimization problems
such as matching, matroid, and stable set problems.
\end{abstract}
\clearpage
\thispagestyle{empty}
\tableofcontents
\clearpage
\setcounter{page}{1}

\section{Introduction}

\subsection{Problem Formulation}\label{sec:formulation}
We study a stochastic variant of linear programming (LP) with the 0/1-integer constraint,
which enables us to discuss such variants of various packing-type combinatorial optimization problems
such as matching, matroid, and stable set problems in a unified manner.
Specifically, we introduce the \emph{stochastic packing integer programming problem} defined as follows:
\begin{align}
\label{eq:packing}
\begin{array}{ll}
\text{maximize} &\ \tilde c^\top x \\[1mm]
\text{subject to} &\ A x \le b, \\[1mm]
&\ x \in \{0, 1\}^m,
\end{array}
\end{align}
where $A \in \mathbb{Z}_+^{n \times m}$ and $b \in \mathbb{Z}_+^n$,
and $\mathbb{Z}_+$ denotes the set of nonnegative integers.
%We assume that $b \ge 1$ without loss of generality,
%and that $Ax \leq b$ and $x \geq 0$ imply $x \leq 1$ for simplicity\footnote{This holds for most of applications,
%and the generalizability to remove this assumption is discussed in Section~\ref{sec:k-CSPIP} with a specific application.}.
The objective vector $\tilde c \in \mathbb{Z}_+^m$ is \emph{stochastic} in the following sense.
\begin{itemize}
\item
  The entries $\tilde c_j$ $(j = 1, 2, \ldots, m)$ are independent random variables
  with some hidden distributions for which we are given the following information: for each $j$,
  \begin{itemize}
  \item
    the domain of $\tilde c_j$ is an integer interval $\{c_j^-, c_j^- + 1, \ldots, c_j^+\}$ given by $c_j^-, c_j^+ \in \mathbb{Z}_+$, and
  \item
    the probability that $\tilde c_j = c_j^+$ is at least a given constant $p \in (0, 1]$ (which is independent from $j$), i.e., $c_j^- \leq \tilde c_j \le c_j^+ - 1$ occurs with probability at most $1 - p$.
  \end{itemize}
\item
  When an instance ($A$, $b$, and the above information on $\tilde c$) is given,
  the {\em realized values} of all $\tilde c_j$, denoted by $c_j$,
  are hiddenly fixed by nature according to the above distributions.
\item
  For each $j$, we are allowed to conduct a query to reveal the realized value $c_j$ of $\tilde c_j$.
\end{itemize}
Note that, since all $\tilde c_j$ are independent,
we can consider at any time that each realized value $c_j$ is determined just when a query for $j$ is conducted. 

\begin{example}\label{ex:SM}
Our problem captures the \emph{stochastic matching problem} introduced by Blum et al.~\cite{blum2015ignorance} as follows.
In the stochastic matching problem, we are given an undirected graph $G = (V, E)$ such that
each edge $e \in E$ is \emph{realized} with probability at least $p \in (0, 1]$,
and the goal is to find a large matching that consists of realized edges.
We can know whether each edge is realized or not by conducting a query.
A naive formulation of this situation as our problem is obtained by restricting the domain of $\tilde{c} \in \mathbb{Z}_+^E$ to $\{0, 1\}^{E}$,
by letting $A \in \mathbb{Z}_+^{V \times E}$ be the vertex-edge incidence matrix of $G$, and by setting $b = 1$.
Section~\ref{sec:nonbipartite} gives a more detailed discussion with general edge weights.
\end{example}

Our aim is to find a feasible solution to \eqref{eq:packing} with a large objective value by conducting a small number of queries.
Note that we can definitely obtain an optimal solution by solving the corresponding non-stochastic problem after conducting queries for all $j$.
Our interest is therefore in the trade-off between the number of queries and the quality of the obtained solution.

\subsection{Our Contributions and Technique}\label{sec:contributions}
\begin{table}[t]
\centering
\caption{Results obtained for the adaptive strategy, where $n$ and $m$ denote the number of vertices in the graph and the ground set size of the matroids (or the number of edges), respectively, in question. We omit $O( \cdot )$ in the iteration column. Also, all the coefficients are assumed to be $O(1)$. For the non-adaptive strategy, the approximation ratio is halved.}
\label{tbl:results}%\hspace{-10mm}
\begin{tabular}{l|c|c} \hline
Problem & Approximation Ratio & Number of Iterations $T$ \\ \hline
Bipartite Matching & $1 - \epsilon$ & $\log (1/\epsilon p) / \epsilon p$ \\
Non-bipartite Matching & $1 - \epsilon$ & $\log (n/\epsilon)/\epsilon p$ \\
$k$-Hypergraph Matching & $(1 - \epsilon) /(k - 1 + 1/k)$ & $(k \log (k/\epsilon p) + 1/\epsilon) / \epsilon p$ \\
$k$-Column Sparse PIP & $(1 - \epsilon)/2k$ & $(k \log (k/\epsilon p) + 1/\epsilon) / \epsilon p$ \\[1mm]
Matroid (Max.~Independent Set)& $1 - \epsilon$ & $\log (m/\epsilon) / \epsilon p$ \\
Matroid Intersection & $1 - \epsilon$ & $\log (m/\epsilon) / \epsilon p$ \\
$k$-Matroid Intersection & $(1 - \epsilon) / k$ & $k \log m \log (m /\epsilon) / \epsilon^3 p$ \\
Matchoid & $(1 - \epsilon) 2/3$ & $\log (m/\epsilon) /\epsilon p$ \\
%$k$-Matroid Matching & $(1 - \epsilon) 2/k$ & $\log (n/\epsilon) /\epsilon p$ \\
Degree Bounded Matroid & $1 - \epsilon$ {\small $\displaystyle\left(\begin{array}{c}\text{each constraint}\\\text{is violated}\\\text{at most $d-1$}\end{array}\right)$}
%$\binom{\text{each constraint is}}{\text{violated at most $d-1$}}$
& $d \log (n/\epsilon)/\epsilon^2 p$ \\[4mm]
Stable Set in Chordal Graphs & $1 - \epsilon$ & $\log n / \epsilon p$ \\
Stable Set in t-Perfect Graphs & $1 - \epsilon$ & $\log n \log (n/\epsilon) /\epsilon^3 p$ \\
\hline
\end{tabular}\vspace{-1mm}
\end{table}

\subsubsection*{Contributions}
We propose a general framework of adaptive and non-adaptive algorithms for the stochastic packing integer programming problem.
Here, an algorithm is \emph{non-adaptive} if it reveals all queried items simultaneously, and \emph{adaptive} otherwise.

In the adaptive strategy\footnote{Algorithms \ref{alg:adaptive} and \ref{alg:nonadaptive} have freedom of the choices of algorithms for solving LPs and for finding an integral solution in the last step; in particular, the latter depends heavily on each specific problem before formulated as an integer LP. For this reason, we use the term ``strategy'' rather than ``algorithm'' to refer them.} (which is formally shown in Algorithm~\ref{alg:adaptive} in Section~\ref{sec:two_strategies}), we iteratively compute an optimal fractional solution $x \in [0,1]^m$ to the \emph{optimistic LP} (the LP relaxation of \eqref{eq:packing} in which all the unrevealed $\tilde c_j$ are supposed to be $c_j^+$), and conduct a query for each element $j$ with probability $x_j$.
After the iterations, we find an integral feasible solution to the \emph{pessimistic LP} (in which all the unrevealed $\tilde c_j$ are supposed to be $c_j^-$) by using some algorithms for the corresponding non-stochastic problem.

Similarly, in the non-adaptive strategy (Algorithm~\ref{alg:nonadaptive}), we iteratively compute an optimal fractional solution $x$ to the optimistic LP, and round down each element $j$ (i.e., suppose $\tilde c_j$ to be $c_j^-$ instead of revealing $c_j$) with probability $x_j$.
After the iterations, we reveal all the rounded-down elements and find an integral feasible solution to the pessimistic LP.

In application, we need to decide how to execute the last step,
and the performance of the resulting algorithm depends on combinatorial structure of each specific problem.
%so our analysis must exploit this structure.
Our main contribution is a \emph{proof technique} for analyzing the performance of the algorithms.
Using this technique, we obtain results for the problem classes summarized in Table~\ref{tbl:results}.

\subsubsection*{Technique}
Our technique is based on \emph{LP duality} and \emph{enumeration}.
A brief overview of the technique follows, where we focus on the adaptive strategy.

Let $\tilde\mu$ be the optimal value of the \emph{omniscient LP} (the LP relaxation of \eqref{eq:packing} in which all $\tilde c_j$ are revealed).
Note that $\tilde\mu$ is a random variable depending on the realization of $\tilde c_j$.
Our goal is to evaluate the number of iterations $T$ such that
the optimal value of the pessimistic LP after $T$ iterations
is at least $(1 - \epsilon) \tilde\mu$ with high probability\footnote{Here we consider two types of randomness together.
One is on the realization of $\tilde c_j$, which is contained in the ``stochastic'' input and determines the omniscient optimal value $\tilde\mu$.
The other is on the choice of queried elements, which is involved in our ``randomized'' algorithms and affects the pessimistic LP obtained after the iterations.}. 
Then, if we have an LP-relative $\alpha$-approximation algorithm~\cite{parekh2014generalized} (which outputs an integral feasible solution whose objective value is at least $\alpha$ times the LP-optimal value) for the corresponding non-stochastic problem,
we obtain a $(1 - \epsilon) \alpha$-approximate solution to our problem with high probability.

%During the iteration, since some $\tilde c_j$ are successively revealed, the optimal value of the pessimistic problem gradually increases.
%Our goal is to prove the following claim: ``After some iterations, the fractional optimal value of the pessimistic problem is at least $(1 - \epsilon) \mu$ with high probability.''

To discuss the optimal value of the pessimistic LP, we consider the dual LP.
By the LP strong duality, it is sufficient to prove that the dual pessimistic LP
after $T$ iterations has no feasible solution whose objective value is less than $(1 - \epsilon) \tilde\mu$ with high probability.

Here, we introduce a finite set $W \subseteq \mathbb{R}_+^n$ of dual vectors,
called a \emph{witness cover},
for every possible objective value $\mu$ (a candidate of $\tilde\mu$)
that satisfies the following property: if all $y \in W$ are infeasible, there is no feasible solution whose objective value is less than $(1 - \epsilon) \mu$.
Intuitively, $W$ represents all the candidates for dual feasible solutions whose objective values are less than $(1 - \epsilon) \mu$. 
We evaluate the probability that each $y \in W$ becomes infeasible after $T$ iterations, and then estimate the sufficient number of iterations by using the union bound for $W$.

In application, we only need to show the existence of a small witness cover for each specific problem.
We also give general techniques to construct small witness covers
when the considered problem enjoys some nice properties,
e.g., when the constraint system $Ax \leq b$, $x \geq 0$ is totally dual integral.

\subsection{Related Work}
As described in Example~\ref{ex:SM},
our stochastic packing integer programming problem generalizes the \emph{stochastic (unweighted) matching problem}~\cite{blum2015ignorance,assadi2016stochastic,assadi2017stochastic} and the \emph{stochastic (unweighted) $k$-hypergraph matching problem}~\cite{blum2015ignorance}, which have recently been studied in EC (Economics and Computation) community.
These problems are motivated to find an optimal strategy for kidney exchange~\cite{roth2004kidney,dickerson2016organ}.

For the stochastic unweighted matching problem, Blum et al.~\cite{blum2015ignorance} proposed adaptive and non-adaptive algorithms that achieve approximation ratios of $(1 - \epsilon)$ and of $(1/2 - \epsilon)$, respectively, \emph{in expectation}, by conducting $O(\log (1/ \epsilon)/p^{2/\epsilon})$ queries per vertex.
Their technique is based on the existence of disjoint short augmenting paths.
Assadi et al.~\cite{assadi2016stochastic} proposed adaptive and non-adaptive algorithms that respectively achieve the same approximation ratios \emph{with high probability}, by conducting $O(\log (1/ \epsilon p)/\epsilon p)$ queries per vertex.
Their technique is based on the Tutte--Berge formula and vertex sparsification.
Our proposed strategies coincide with those of Assadi et al.\ when they are applied to the stochastic unweighted matching problem and we always find integral optimal solutions to the LP relaxations, i.e.,
solve the (non-stochastic) unweighted matching problem every time.
Our analysis looks similar to theirs since they both use the duality,
but ours is simpler and can also be used for the weighted and capacitated situation.
On the other hand, our analysis shows that $O(\log (n/\epsilon)/\epsilon p)$ queries per vertex are required\footnote{Very recently, Behnezhad and Reyhani~\cite{behnezhad2017almost} claimed that the same algorithm as ours achieves an approximation ratio of $1 - \epsilon$ by conducting a constant number of queries that depends on only $\epsilon$ and $p$. 
Their analysis uses augmenting paths, like Blum et al.~\cite{blum2015ignorance}.}, which is worse than theirs. 

Recently, Assadi et al.~\cite{assadi2017stochastic} proposed a non-adaptive algorithm that achieves an approximation ratio of strictly better than $1/2$ \emph{in expectation}. 
However, this technique is tailored to the unweighted matching problem, so we could not generalize it to our problem.

For the stochastic unweighted $k$-hypergraph matching problem, Blum et al.~\cite{blum2015ignorance} proposed adaptive and non-adaptive algorithms that find $(2 - \epsilon)/k$- and $(4 - \epsilon)/(k^2 + 2k)$-approximate matchings, respectively, \emph{in expectation}, by conducting $O(s_{k,\epsilon} \log (1/ \epsilon)/p^{s_{k,\epsilon}})$ queries per vertex, where $s_{k,\epsilon}$ is a constant depending on $k$ and $\epsilon$.
Their technique is based on the local search method of Hurkens and Schrijver~\cite{hurkens1989size}.
For the adaptive case, our strategy achieves a worse approximation ratio than theirs
because the same is true of the LP-based algorithm versus the local search. 
On the other hand, our algorithm requires an exponentially smaller number of queries and runs in polynomial time both in $n$ and $1/\epsilon$.
In addition, our algorithm can be used for the weighted case.
For the non-adaptive case, our algorithm outperforms theirs,
all in terms of approximation ratio, the number of queries, and running time.

Other variants of the stochastic packing integer programming problem with queries have been studied.
However, many of them employ the \emph{query-commit model}~\cite{dean2004approximating,dean2005adaptivity,molinaro2011query,costello2012stochastic}, in which the queried elements must be a part of the output.
Some studies~\cite{adamczyk2011improved,chen2009approximating,bansal2012lp} also impose additional budget constraints on the number of queries.
In the \emph{stochastic probing problem}~\cite{gupta2013stochastic,adamczyk2016submodular,gupta2017adaptivity}, both the queried and realized elements must satisfy given constraints.
Blum et al.~\cite{blum2013harnessing} studied a stochastic matching problem without query-commit condition, but with a budget constraint on the number of queries.

\subsection{Organization}

The rest of the paper is organized as follows.
In Section~\ref{sec:general}, we describe our framework of adaptive and non-adaptive algorithms for the stochastic packing integer programming problem,
and explain a general technique for providing a bound on the number of iterations.
In Section~\ref{sec:witness}, we outline how to construct a small witness cover in general.
In Section~\ref{sec:applications}, we apply the technique to a variety of specific combinatorial problems.
In Section~\ref{sec:sparsification}, we provide a vertex sparsification lemma that can be used to improve the performance of the algorithms for several problems.

\section{General Framework}
\label{sec:general}

Throughout the paper (with one exception as remarked later),
we assume that the constraints in \eqref{eq:packing} satisfy several reasonable conditions.

\begin{assumption}\label{asmp:constraint}
  We assume that $A \in \mathbb{Z}_+^{n \times m}$ and $b \in \mathbb{Z}_+^n$ in \eqref{eq:packing} satisfy the following three conditions\footnote{The first two are assumed without loss of generality (by removing the corresponding constraints and variables if violated).
    The third one is for simplicity, which holds for most of applications.
    The generalizability to remove it is discussed in Section~\ref{sec:k-CSPIP} with a specific application.}:
  \begin{enumerate}
    \renewcommand{\labelenumi}{\alph{enumi}.}
    \item
      $b \geq 1$;
    \item
      $A\chi_j \leq b$ for each $j = 1, 2, \ldots, m$, where $\chi_j \in \{0, 1\}^m$ denotes the $j$-th unit vector;
    \item
      $Ax \leq b$ and $x \geq 0$ imply $x \leq 1$.
  \end{enumerate}
\end{assumption}

%This section is organized as follows.
We give a general framework of adaptive and non-adaptive algorithms for our problem in Section~\ref{sec:two_strategies},
and then describe a unified methodology for its performance analysis in Section~\ref{sec:performance_analysis}.
The main results are stated as Theorems \ref{thm:adaptive} and \ref{thm:nonadaptive}, whose proofs are separately shown in Section~\ref{sec:proofs}.

\subsection{Two Strategies}\label{sec:two_strategies}

To describe two strategies,
we formally define two auxiliary problems, the \emph{optimistic LP} and the \emph{pessimistic LP}.
We define the \emph{optimistic vector} $\overline{c} \in \mathbb{Z}_+^m$ and the \emph{pessimistic vector} $\underline{c} \in \mathbb{Z}_+^m$ as follows:
\begin{align}
  \overline{c}_j = \begin{cases} c_j & j \text{ has been queried}, \\ c_j^+ & \text{otherwise}, \end{cases}\quad
  \underline{c}_j = \begin{cases} c_j & j \text{ has been queried}, \\ c_j^- & \text{otherwise},  \end{cases}
\end{align}
where recall that $c_j \in \mathbb{Z}_+$ denotes the realized value of $\tilde c_j$.
The optimistic and pessimistic LPs are obtained from the original stochastic problem \eqref{eq:packing}
by replacing the objective vector $\tilde c$ with $\overline{c}$ and with $\underline{c}$, respectively,
and by relaxing the constraint $x \in \{0, 1\}^m$ to $x \in \mathbb{R}_+^m$,
where $\mathbb{R}_+$ denotes the set of nonnegative reals.
By Assumption~\ref{asmp:constraint}.c ($Ax \leq b$ and $x \geq 0$ imply $x \leq 1$),
the relaxed constraint is equivalent to $x \in [0, 1]^m$.
Note that these problems are no longer stochastic, i.e., contain no random variables.

\begin{algorithm}[t]
\caption{Adaptive strategy.}
\label{alg:adaptive}
\begin{algorithmic}[1]
\For{$t = 1, 2, \ldots, T$}\label{line:1}
\State{Find an optimal solution $x$ to the optimistic LP.} \label{line:2}
\State{For each $j = 1, \ldots, m$, conduct a query to reveal $\tilde c_j$ with probability $x_j$.}\label{line:3}
\EndFor
\State{Find an integral feasible solution to the pessimistic LP and return it.}\label{line:5} 
\end{algorithmic}
\end{algorithm}

\begin{algorithm}[t]
\caption{Non-adaptive strategy.}
\label{alg:nonadaptive}
\begin{algorithmic}[1]
\For{$t = 1, 2, \ldots, T$}
\State{Find an optimal solution $x$ to the optimistic LP.}\label{line:2'}
\State{For each $j = 1, \ldots, m$, suppose $\tilde c_j = c_j^-$ with probability $x_j$.}\label{line:3'}
\EndFor
\State{For every $j$ with $\tilde c_j = c_j^-$ supposed at Line~\ref{line:3'}, conduct a query to reveal $\tilde c_j$.}
\State{Find an integral feasible solution to the pessimistic LP and return it.}\label{line:6}
\end{algorithmic}
\end{algorithm}

First, we describe the {\em adaptive strategy} shown in Algorithm~\ref{alg:adaptive}.
In this strategy, we iteratively compute an optimal solution $x \in [0, 1]^m$ to the optimistic LP\footnote{\label{ft:5}Note that,
if the optimal solution $x$ is written as a convex combination $\sum_{i} \lambda_i x^{(i)}$ of basic feasible solutions $x^{(i)}$,
then every $x^{(i)}$ is also optimal and one can replace $x$ with any $x^{(i)}$.
In particular, when the considered polyhedron is integral (i.e., every extreme point is an integral vector),
Algorithms~\ref{alg:adaptive} and \ref{alg:nonadaptive} can be derandomized based on this observation.},
and reveal each $\tilde c_j$ with probability $x_j$.
After $T$ iterations, we find an integral feasible solution to the pessimistic LP,
where we have freedom of the choice of algorithms for the corresponding non-stochastic problem.
As remarked in Section~\ref{sec:contributions}, how to execute the last step depends heavily on each specific problem.

Next, we describe the {\em non-adaptive strategy} shown in Algorithm~\ref{alg:nonadaptive}.
As with the adaptive strategy, we solve the optimistic LP at each step.
To be non-adaptive, the algorithm tentatively assigns values to $\tilde c_j$ pessimistically instead of revealing their realized values.
After the iterations, it reveals all these values and then computes an integral feasible solution to the pessimistic LP by some algorithms for the non-stochastic problem.

\subsection{Performance Analysis}\label{sec:performance_analysis}

We now analyze the performance of algorithms within our framework.
First, we consider the adaptive strategy (Algorithm~\ref{alg:adaptive}).
As described in Section~\ref{sec:formulation},
we evaluate the trade-off between the following two factors,
each of which is reasonably decomposed into two factors.
\begin{itemize}
\item[(1)]
  {\bf The number of conducted queries.}
  \begin{itemize}%[\quad (1)]
  \item[(1-a)]
    \emph{Expected number of queries at Line~\ref{line:3}}.
    If this number is large, the algorithm may reveal all relevant $\tilde c_j$ in a few iterations, making the algorithm trivial.

  \item[(1-b)]
    \emph{Required number of iterations $T$ at Line~\ref{line:1}}.
    If $T$ is very large, then, as in (1-a), the algorithm may reveal all relevant $\tilde c_j$, making the algorithm trivial.
  \end{itemize}
\item[(2)]
  {\bf The quality of the output solution.}
  Basically, we want to find a feasible solution to \eqref{eq:packing} with a large objective value,
  which is at most the {\em omniscient} optimal value of \eqref{eq:packing} after all $\tilde c_j$ are revealed.
  \begin{itemize}
  \item[(2-a)]
    \emph{Closeness between the pessimistic and omniscient LPs}.
    The omniscient optimal value of \eqref{eq:packing} is at most the optimal value $\tilde\mu$ of the {\em omniscient LP},
    which is obtained by revealing all $\tilde c_j$ and by relaxing $x \in \{0, 1\}^m$ to $x \in \mathbb{R}_+^m$.
    If the pessimistic LP-optimal value at Line~\ref{line:5} is close to $\tilde\mu$,
    then, at least as an LP, the pessimistic problem is close to the problem that we want to solve.
  \item[(2-b)]
    \emph{LP-relative approximation ratio at Line~\ref{line:5}}.
    If one can find an integral feasible solution such that
    the ratio between its objective value and the LP-optimal value is bounded, then, combined with (2-a),
    a reasonable bound on the objective value of the output solution can be obtained.
    %We only prove that the optimal values of the pessimistic and omniscient LPs are almost the same.
    %This \emph{does not} guarantee that this is also true of the integral optimal values.
  \end{itemize}
\end{itemize}

Essentially, (1-a) and (2-b) are properties of each specific problem and its LP formulation.
Thus, we postpone these two factors to the discussion on applications in Section~\ref{sec:applications},
and focus on (1-b) and (2-a) in the general study in this section.
%
%Let $\mu$ be the optimal value of the \emph{omniscient LP},
%which is obtained by revealing all $\tilde c_j$ and by relaxing $x \in \{0, 1\}^m$ to $x \in [0, 1]^m$.
%Note that $\mu$ is a random variable depending on the realization of $\tilde c_j$.
That is, our goal here is to estimate $T$ such that the optimal value of the pessimistic LP after $T$ iterations
is at least $(1 - \epsilon) \tilde\mu$ with high probability,
where $\tilde\mu$ is the optimal value of the omniscient LP and $\epsilon > 0$ is a parameter one can choose.
Note again that $\tilde\mu$ is a random variable depending on the realization of $\tilde c_j$.

To evaluate the number of iterations $T$, we consider the dual of the pessimistic LP:
\begin{align}
\label{eq:packingdual}
\begin{array}{ll}
\text{minimize} &\ y^\top b \\[1mm]
\text{subject to} &\ y^\top A \ge \underline{c}^\top\!, \\[1mm] 
&\ y \in \mathbb{R}_+^n.
\end{array}
\end{align}
By the LP strong duality, it is sufficient to evaluate the probability that this dual LP has no feasible solution whose objective value is less than $(1 - \epsilon) \tilde\mu$.

Now we introduce the notion of a \emph{witness cover}, which is the most important concept in this study.
Intuitively, a witness cover for $\mu \in \mathbb{R}_+$
is a set of ``representatives'' of all the dual feasible solutions
with objective values of at most $(1 - \epsilon) \mu$.
More specifically, for any primal objective vector, if some dual feasible solution has the objective value at most $(1 - \epsilon)\mu$,
then a witness cover contains at least one such dual solution. 

\begin{definition}\label{def:witness}
Let $A \in \mathbb{Z}_+^{n\times m}$, $b \in \mathbb{Z}_+^n$, and $\epsilon, \epsilon' \in \mathbb{R}_+$ with $0 < \epsilon' \leq \epsilon$.
A finite set $W \subseteq \mathbb{R}_+^n$ of dual vectors is an \emph{$(\epsilon,\epsilon')$-witness cover for $\mu \in \mathbb{R}_+$} if it satisfies the following two properties.
\begin{enumerate}
\item For every $c \in \mathbb{Z}_+^m$,
  if $y^\top A \geq c^\top$ is violated (i.e., $(y^\top A)_j < c_j$ for some $j$) for all $y \in W$,
  then $y^\top A \geq c^\top$ is violated for all $y \in \mathbb{R}_+^n$ with $y^\top b \leq (1 - \epsilon)\mu$.
\item $y^\top b \leq (1 - \epsilon') \mu$ holds for all $y \in W$.
\end{enumerate}
\end{definition}

\begin{example}
Consider the bipartite matching case (see Section~\ref{sec:bipartite} for the detail).
In the LP relaxation of the naive formulation \eqref{eq:bipartite},
the constraint system is totally dual integral (see Section~\ref{sec:tdipoly} for the detail),
and each dual vector is an assignment of nonnegative reals to vertices, whose sum is the objective value.
Hence, the set of assignments of nonnegative integers to vertices whose sum is at most $(1 - \epsilon)\mu$
is an $(\epsilon, \epsilon)$-witness cover for $\mu$. 
\end{example}

During the iterations, the constraints in the dual pessimistic LP \eqref{eq:packingdual} become successively stronger.
Hence, for any witness cover $W$ for the omniscent LP-optimal value $\tilde\mu$, every $y \in W$ eventually becomes infeasible to \eqref{eq:packingdual} (by the second condition in Definition~\ref{def:witness}).
By evaluating the probability that all $y \in W$ become infeasible after $T$ iterations, we obtain a bound on the required number of iterations.
Note again that $\tilde\mu$ is a random variable, and hence
we assume that there exists a relatively small witness cover
for every possible objective value $\mu$,
which can be restricted to $\mu \geq 1$ due to Assumption~\ref{asmp:constraint}.b
(see the proof for the detail).

\begin{theorem}
\label{thm:adaptive}
Let $M \in \mathbb{R}_+$,
and suppose that there exists an $(\epsilon, \epsilon')$-witness cover of size at most $M^\mu$ for every $\mu \geq 1$.
Then, by taking
\begin{align}\label{eq:iterations}
  T \geq \frac{\Delta_c}{\epsilon' p}\log\left(\frac{M}{\delta}\right),
\end{align}
the pessimistic LP at Line~\ref{line:5} of Algorithm~\ref{alg:adaptive} has a $(1 - \epsilon)$-approximate solution with probability at least $1 - \delta$, where $\Delta_c = \max_j (c_j^+ - c_j^-)$ and $0 < \delta < 1$.
\end{theorem}

For the non-adaptive algorithm (Algorithm~\ref{alg:nonadaptive}), by conducting a similar analysis with a case analysis, we obtain the required number of iterations with a provable approximation ratio.

\begin{theorem}
\label{thm:nonadaptive}
Under the same assumption as Theorem~\ref{thm:adaptive},
by taking $T$\! as \eqref{eq:iterations},
the pessimistic LP at Line~\ref{line:6} of Algorithm~\ref{alg:nonadaptive} has a $(1 - \epsilon)/2$-approximate solution with probability at least $1 - \delta$.
\end{theorem}

These theorems show that if there exists a small witness cover (for each possible $\mu$),
Algorithms~\ref{alg:adaptive} and \ref{alg:nonadaptive} will find good solutions in a reasonable number of iterations.
It is worth emphasizing that we only have to prove the existence of such a witness cover, i.e., we do not have to construct it algorithmically.
We discuss how to prove the existence of such witness covers
(i.e., how to construct them theoretically)
in general and in each specific application,
in Sections~\ref{sec:witness} and \ref{sec:applications}, respectively.

\subsection{Proofs of Main Theorems}\label{sec:proofs}
\subsubsection*{Proof of Theorem~\ref{thm:adaptive}}
Let $\tilde\mu$ be the optimal value of the omniscient LP.
If $\tilde\mu = 0$ then the statement obviously holds (with probability 1). 
Thus we restrict ourselves to the case when $\tilde\mu > 0$.
Note that $\tilde \mu > 0$ implies $\tilde\mu \geq 1$ as follows.
If $\tilde \mu > 0$ then $\tilde c_j = c_j \geq 1$ for some $j$, and by Assumption~\ref{asmp:constraint}.b, the $j$-th unit vector $\chi_j \in \{0, 1\}^m$ is feasible
(i.e., $A \chi_j \leq b$); therefore $\tilde\mu \ge c^\top \chi_j = c_j \geq 1$.

For each $\mu \geq 1$, fix an $(\epsilon, \epsilon')$-witness cover $W_\mu$ of size $|W_\mu| \le M^\mu$.
We first evaluate the probability that each $y \in W_{\tilde\mu}$ is feasible after $T$ iterations.
Since some $\tilde c_j$ is newly revealed,
some constraints may be violated (i.e., $(y^\top A)_j < c_j$ may happen).
Once $y$ has become infeasible, it never returns to feasible due to the monotonicity of $\underline{c}$ throughout Algorithm~\ref{alg:adaptive}. 
Therefore, $y$ is feasible after $T$ iterations only if
$y$ is feasible at every iteration step.

Fix $t = 1, 2, \ldots, T$, and we evaluate the probability that
a vector $y$ in each witness cover that is feasible at the beginning of the $t$-th step
remains feasible at the end of the step.
Let $\overline{c}, \underline{c} \in \mathbb{Z}_+^m$ be
the optimistic and pessimistic vectors, respectively, at Line~\ref{line:2} in the $t$-th step,
and $\overline{\mu}, \underline{\mu} \in \mathbb{R}_+^m$
the optimal values of the corresponding LPs.
Note that $\overline{\mu}$ and $\underline{\mu}$ are respectively upper and lower bounds on $\tilde\mu$ at that time.

\begin{Claim}
\label{lem:feasibility}
For every $\mu \in [\underline{\mu}, \overline{\mu}]$ and each $y \in W_{\mu}$ with $y^\top A \geq \underline{c}^\top$\!,
the probability that $y$ is feasible after Line~\ref{line:3} is at most $\exp \left( - \epsilon' p\mu/\Delta_c \right)$.
\end{Claim}
\begin{proof}
Since $y$ is feasible at the beginning of the step,
$y$ is feasible after Line~\ref{line:3} only if no possibly violated constraint is revealed to be $c_j^+$.
We can evaluate the number of possibly violated constraints at this step using the following inequality:
\begin{align}\label{eq:c-x}
	\overline{c}^\top x \le \overline{c}^\top x + y^\top (b - A x) = y^\top b + (\overline{c}^\top - y^\top A) x,
\end{align}
where $x \in [0, 1]^m$ is the optimal solution to the optimistic LP obtained in Line~\ref{line:2}.
Since the optimistic vector $\overline{c}$ dominates the actual vector $\tilde c$
(irrespective of which values are realized),
we have $\overline{c}^\top x \ge \tilde\mu$.
Since $y \in W_{\mu}$, we have $y^\top b \le (1 - \epsilon') \mu$.
Therefore, we derive from \eqref{eq:c-x}
\begin{align}\label{eq:delta-mu}
	\epsilon'\mu \le (\overline{c}^\top - y^\top A)x %\\\nonumber
        \le \sum_{j\colon \text{violated}} (c_j^+ - \underline{c}_j) x_j %\\\nonumber
        \le \Delta_c \sum_{j\colon \text{violated}} x_j,
\end{align}
where we say that $j$ is {\em violated} if $(y^\top A)_j < c_j$ for the realized value $c_j$ of $\tilde c_j$,
and note that $\overline{c}_j \leq c_j^+$, $(y^\top A)_j \geq \underline{c}_j$, and $c_j \geq c_j^+ - \Delta_c$ for every $j$.
Since the left-hand side of \eqref{eq:delta-mu} is positive, there must exist possibly violated constraints in the support of $x$, and if one of them, say $\tilde c_j$, is revealed (with probability $x_j$) as $c_j^+$ (with probability at least $p$), then $y$ becomes infeasible.
%To evaluate this probability, let us define
%\begin{align}
%	X_j &= \begin{cases} 1 & \text{with probability } x_j, \\ 0 & \text{otherwise}. \end{cases}
%\end{align}
Then the probability that $y$ is still feasible after this step is at most
\begin{align}
	%\E\left[ \prod_{j\colon \text{violated}} (1 - p X_j) \right] =
        \quad \prod_{j\colon \text{violated}} (1 - p x_j) &\le \exp \left( -p \sum_{j\colon \text{violated}} x_j \right) %\\\nonumber
        \le \exp \left( \frac{- p \epsilon'\mu}{\Delta_c} \right). \quad \qedhere
\end{align}
\end{proof}

By applying Claim~\ref{lem:feasibility} to $\tilde\mu$ (the omniscient LP-optimal value) $T$ times,
we obtain that the probability that each $y \in W_{\tilde\mu}$ is feasible after $T$ iterations
is at most $\exp \left( - \epsilon' p\tilde\mu T/\Delta_c \right)$.
By the union bound, the probability that $W_{\tilde\mu}$ has at least one feasible solution to the dual pessimistic LP \eqref{eq:packingdual} after $T$ iterations is at most $|W_{\tilde\mu}| \exp \left( - \epsilon' p\tilde\mu T / \Delta_c \right)$,
which is at most $\exp \left( \tilde\mu \log M - \epsilon' p\tilde\mu T / \Delta_c \right)$. 
By taking $T \ge \Delta_c \log (M/\delta) / \epsilon' p$,
the latter value is bounded by $\exp(\tilde\mu \log\delta) \leq \delta$
(recall that $\tilde\mu \geq 1$ and $0 < \delta < 1$).
By the definition of witness cover and strong duality,
we conclude that the optimal value of the pessimistic LP at Line~\ref{line:5} of Algorithm~\ref{alg:adaptive} is at least $(1 - \epsilon) \tilde\mu$ with probability at least $1 - \delta$.

\subsubsection*{Proof of Theorem~\ref{thm:nonadaptive}}

The following proof is a simple extension of Theorem~5.1 in Assadi et al.~\cite{assadi2016stochastic} for the stochastic matching problem.

Let $\tilde\mu$ be the optimal value of the omniscient LP,
and we assume $\tilde\mu \geq 1$ as in the proof of Theorem~\ref{thm:adaptive}.
In the above analysis of Algorithm~\ref{alg:adaptive}, it is ensured that there exists a solution $x$ with $\overline{c}^\top x \ge \tilde\mu$, which ensured that the last pessimistic LP in Algorithm~\ref{alg:adaptive} has an optimal value of at least $(1 - \epsilon) \tilde\mu$. 
However, in the non-adaptive case, we may not be able to find such a solution because each $\tilde c_j$ is not revealed but is rounded-down.

To overcome this issue, we define $\tilde\mu' \in \mathbb{R}_+$ as the minimum objective value obtained at Line~\ref{line:2'} of Algorithm~\ref{alg:nonadaptive}.
Note that, since the optimal value of the optimistic LP solved at Line~\ref{line:2'} is monotonically non-increasing, $\tilde\mu'$ is the objective value obtained at the $T$-th step.
By applying Claim~\ref{lem:feasibility} to $\tilde\mu'$ (instead of $\tilde\mu$) $T$ times, we obtain the following claim.
\begin{Claim}
\label{lem:largecase}
%Suppose that there exists an $(\epsilon,\delta)$-witness cover of size at most $M^{\mu'}$ for $\mu'$.
By taking $T \ge \Delta_c \log (M/\delta) / \epsilon' p$ in Algorithm~\ref{alg:nonadaptive}, the optimal value of the pessimistic LP at Line~\ref{line:6} is at least $(1 - \epsilon) \tilde\mu'$ with probability at least $1 - \delta$. 
%\qed
\end{Claim}
If $\tilde\mu' \ge \tilde\mu/2$, we can immediately prove the theorem.
Thus, we consider the case $\tilde\mu' < \tilde\mu/2$, obtaining the following claim.
\begin{Claim}
\label{lem:smallcase}
If $\tilde\mu' < \tilde\mu/2$, the optimal value of the pessimistic LP at Line~\ref{line:6} of Algorithm~\ref{alg:nonadaptive} is at least $\tilde\mu/2$.
\end{Claim}
\begin{proof}
We use the subscripts $R$ and $N$ to denote the revealed and unrevealed entries in the primal vector, respectively,
i.e., $\tilde c_R = c_R$ has been realized and the rest $\tilde c_N$ has not been revealed.
Let $x^* \in \mathbb{R}_+^m$ be an optimal solution to the (primal) omniscient LP
(which is a random variable depending on the realization of $\tilde c_N$). We then have
\begin{align}
	c_R^\top x_R^* + \tilde c_N^{\top} x_N^* = \tilde\mu,
\end{align}
%where $c_R$ denotes the realization of $\tilde c_R$. %with all entries revealed.
Since $(0, x_N^*) \in \mathbb{R}_+^m$ is a feasible solution to the optimistic LP at the last iteration, we have
\begin{align}
	c_N^{+\top} x_N^* \le \tilde\mu' < \tilde\mu / 2.
\end{align}
Therefore, the objective value for $x^*$ in the pessimistic LP at Line~\ref{line:6} is bounded by
\begin{align}
	\underline{c}^\top x^* \ge c_R^\top x_R^* \ge 
	c_R^\top x_R^* + (\tilde c_N^{\top} - c_N^{+\top}) x_N^* > \tilde\mu / 2.
\end{align}
%where $\tilde c_N \le c_N^{+}$ is used.
This means that the pessimistic LP-optimal value is at least $\tilde\mu / 2$.
\end{proof}
This concludes the theorem.

\section{Constructing Witness Covers}
\label{sec:witness}

Our technique requires us to prove the existence of a small witness cover.
Here, we describe general strategies for constructing small witness covers.

\subsection{Totally Dual Integral Case}
\label{sec:tdipoly}

A system $A x \le b$, $x \ge 0$ is \emph{totally dual integral (TDI)} if,
for every integral objective vector $c \in \mathbb{Z}^m$,
the dual problem $\min \{\, y^\top b : y^\top A \ge c^\top,~ y \ge 0 \,\}$ has an integral optimal solution $y \in \mathbb{Z}_+^n$ (unless it is infeasible).
Note that every TDI system yields an integral polyhedron
(see, e.g., \cite{schrijver2003combinatorial} for the detail).
Hence, if we obtain a basic optimal solution to the optimistic LP (in Line~\ref{line:2} of Algorithms~\ref{alg:adaptive} and \ref{alg:nonadaptive}),
then we do not need randomization in conducting query (cf.~the footnote~\ref{ft:5} in Section~\ref{sec:two_strategies}).

If the system is TDI, we can construct a witness cover by enumerating all possible integral dual vectors as follows.
\begin{lemma}\label{lem:tdi}
If the system $Ax \le b$, $x \ge 0$ is TDI,
the following set $W \subseteq \mathbb{R}_+^n$ is an $(\epsilon, \epsilon)$-witness cover for $\mu \geq 1$ such that $|W| = \exp\left({O\left(\mu \log (1 + \frac{n}{\mu})\right)}\right)$:
\begin{align}
\label{eq:tdipoly}
	W = \{\, y \in \mathbb{Z}_+^n : y^\top b \le (1 - \epsilon) \mu \,\}.
\end{align}
\end{lemma}
\begin{proof}
It is clear that $W$ is an $(\epsilon, \epsilon)$-witness cover for $\mu$,
so it only remains to evaluate the cardinality of $W$.
%Since $b \ge 1$, we have $\sum_i y_i < \mu$.
%If $\mu \leq 1$, this holds only for $y = 0$ and hence $|W| = 1$.
%Otherwise (i.e., if $\mu > 1$),
We see that $|W|$ is at most the number of nonnegative vectors whose entries sum to
at most $\floor{\mu}$,
which can be counted by distributing $k$ ($\leq \floor{\mu}$) tokens among $n$ entries, giving
\begin{align}
	\quad |W| &\le \sum_{k=0}^{\floor{\mu}} \binom{n + k - 1}{k} = \binom{n + \floor{\mu}}{\floor{\mu}}
        \le \left( \frac{e (n + \floor{\mu})}{\floor{\mu}} \right)^{\floor{\mu}}
        = e^{O\left( \mu \log \left( 1 + \frac{n}{\mu} \right) \right)}. \quad \qedhere
\end{align}
\end{proof}

Note that the same counting technique can be used when the system is totally dual $1/k$-integral (TDI/$k$), i.e., the existence of a dual optimal solution where each entry is a multiple of $1/k$ is guaranteed.

\subsection{Non-TDI Case}
\label{sec:nontdi}

If the system is not TDI, we have to deal with fractional dual vectors.
To enumerate these fractional vectors, we discretize the dual vectors, requiring the discretization to have the following property: 
if there exists a feasible $y$ such that $y^\top b \le (1 - \epsilon) \mu$, there exists a feasible discretized $y'$ such that $y'^{\top} b \le (1 - \epsilon/2) \mu$. 
Here, we consider two possible situations: Dual Sparse Case and  General Case.

\subsubsection*{Dual Sparse Case}
%\paragraph{Dual Sparse Case}

If there exists a sparse dual optimal solution, we can simply discretize the dual vectors to obtain a good discretized solution as follows.

\begin{lemma}
\label{lem:nontdi_sparse}
For positive $\mu$, $\epsilon$, and $\gamma$,
if there exists $y \in \mathbb{R}_+^n$ such that $y^\top b \le (1 - \epsilon) \mu$ and $|\supp(y)| \le \gamma \mu$,
then there exists $y' \in \prod_{i = 1}^n (\epsilon/2 b_i \gamma) \mathbb{Z}_+$ such that $y'^\top A \ge y^\top A$ and $y'^\top b \le (1 - \epsilon/2) \mu$.
\end{lemma}
\begin{proof}
A suitable $y'$ can be obtained by rounding up the $i$-th entry of $y$ to the next multiple of $\epsilon / 2 b_i \gamma$ for each $i$.
\end{proof}

Now that the existence of a discretized solution whose objective value is almost the same as any sparse solution has been guaranteed,
we can construct a witness cover by enumerating all the discretized vectors. %as follows.
\begin{lemma}
\label{lem:poly_nontdi_sparse}
Under the assumption given in Lemma~\ref{lem:nontdi_sparse}, the following set $W \subseteq \mathbb{R}_+^n$ is an $(\epsilon, \epsilon/2)$-witness cover for $\mu \geq 1$, whose cardinality is $|W| = \exp\left({ O\left(\mu \left( \gamma \log \frac{n}{\gamma \mu} + \frac{1}{\epsilon} \right) \right)}\right)$:
\begin{align}
	W = \left\{\, y \in \prod_{i=1}^n \left(\frac{\epsilon}{2 b_i \gamma}\right) \mathbb{Z}_+ : \, y^\top b \le \left(1 - \frac{\epsilon}{2}\right) \mu,~ |\supp(y)| \le \gamma \mu \,\right\}.
\end{align}
\end{lemma}
\begin{proof}
By Lemma~\ref{lem:nontdi_sparse}, $W$ is an $(\epsilon, \epsilon/2)$-witness cover for $\mu$.
We evaluate the cardinality of $W$ as follows.
We first select $s$ ($\le \gamma \mu$) entries for the support of $y$, and then distribute $k$ ($< 2 \mu / \epsilon$) tokens among these entries, where each token contributes $\epsilon/2$ to the objective value.
In the nontrivial case when $\mu_1 := \gamma\mu \geq 1$ and $\mu_2 := 2\mu/\epsilon \geq 1$,
the number of these patterns is bounded by
\begin{align}
	\quad |W| &\le \sum_{s=1}^{\floor{\mu_1}} \binom{n}{s} \sum_{k=0}^{\floor{\mu_2}} \binom{s + k - 1}{k} \notag \\
        &\le \floor{\mu_1} \left(\frac{e n}{\floor{\mu_1}}\right)^{\floor{\mu_1}} \floor{\mu_2} \left( \frac{e (\floor{\mu_1} + \floor{\mu_2})}{\floor{\mu_2}} \right)^{\floor{\mu_2}} \notag \\\nonumber
        &\le \floor{\mu_1} \left(\frac{e n}{\floor{\mu_1}}\right)^{\floor{\mu_1}} \floor{\mu_2} e^{\floor{\mu_1} + \floor{\mu_2}}
        = \exp \left(O \left( \mu_1 \log \frac{n}{\mu_1} + \mu_2 \right) \right). \quad \qedhere
\end{align}
\end{proof}

\subsubsection*{General Case}
%\paragraph{General Case}

When the optimal dual solutions are not sparse, the results of simple discretization are useless.
However, even in such a case, there is a good discretized solution. % whose entries are multiples of some constant.
Let $y$ be a feasible dual vector. 
Then by applying \emph{randomized rounding}~\cite{raghavan1987randomized} to $y$, we obtain a suitable discretized vector $y'$ with a positive probability.
Formally, the following theoretical guarantee is obtained.
\begin{theorem}[Kolliopoulos and Young~\cite{kolliopoulos2005approximation}]
\label{thm:kolliopoulos2005}
Any feasible LP $\min_y\{\, y^\top b : y^\top A \ge c^\top,\ y \ge 0 \,\}$
has a $(1 + \epsilon/2)$-approximate solution whose entries are multiples of $\theta = \Theta\left(\frac{\epsilon^2}{\log m}\right)$, where $m$ is the dimension of $c$.
\end{theorem}
We use this theorem as an existence theorem.
If there is an optimal dual solution with objective value of at most $(1 - \epsilon) \mu$, this theorem shows that there exists a dual feasible solution whose entries are multiples of $\theta$ with an objective value of at most $(1 + \epsilon/2)(1 - \epsilon) \mu \le (1 - \epsilon/2) \mu$.
By enumerating the dual vectors whose entries are multiples of $\theta$, we can obtain a witness cover.
\begin{lemma}
\label{lem:poly_nontdi_nonsparse}
Let $\theta = \Theta\left(\frac{\epsilon^2}{\log m}\right)$.
The following set $W \subseteq \mathbb{R}_+^n$ is an $(\epsilon, \epsilon/2)$-witness cover for $\mu \geq 1$ such that $|W| = \exp\left({O\left(\frac{\mu\log m}{\epsilon^2} \log \left(1 + \frac{n}{\mu} \right)\right)}\right) $:
\begin{align}
	W = \left\{\, y \in \theta \mathbb{Z}_+^n :\, y^\top b \le \left(1 - \frac{\epsilon}{2}\right) \mu \,\right\}.
\end{align}
\end{lemma}
\begin{proof}
By Theorem~\ref{thm:kolliopoulos2005}, $W$ is an $(\epsilon, \epsilon/2)$-witness cover for $\mu$. 
We evaluate the cardinality of $W$ as follows. Let $\mu' := \floor{\mu/\theta}$.
The number of ways of distributing $k$ ($\leq \mu'$) tokens among $n$ entries is bounded by
\begin{align}
  \quad |W| &\le \sum_{k=0}^{\mu'} \binom{n+k-1}{k}
  \le \mu' \left( \frac{e(n + \mu')}{\mu'} \right)^{\mu'} \nonumber\\ 
  &\le \mu' e^{\mu' (1 + \log (1 + n/\mu'))} = \exp\left(O \left(\frac{\mu}{\theta} \log \left(1 + \frac{\theta n}{\mu} \right) \right) \right). \quad \qedhere
\end{align}

\end{proof}

\subsection{Exponentially Many Constraints}
\label{sec:tdiexp}

Some problems, such as the non-bipartite matching problem and matroid problems, have exponentially many constraints.
In such cases, it is impossible to enumerate all the candidates naively as we have done in the previous sections.

Sometimes, this difficulty is overcome by identifying the granularity of the dual solution (i.e., TDI, dual sparse, or general), and then bounding the number of possible dual patterns by exploiting the combinatorial structure; see next section for concrete examples.

\section{Applications}
\label{sec:applications}

In this section, we demonstrate our proposed framework by applying it to several stochastic combinatorial problems.
We only describe the results for the adaptive strategy (Algorithm~\ref{alg:adaptive}) as the results for the non-adaptive strategy (Algorithm~\ref{alg:nonadaptive}) can easily be obtained analogously.

\subsection{Matching Problems}\label{sec:matching_problems}
In this section, unless otherwise noted, $n$ denotes the number of vertices in the (hyper)graph in question.

\subsubsection{Bipartite Matching}\label{sec:bipartite}

We first demonstrate how to use our technique for the bipartite matching problem.
Let $(V, E)$ be a bipartite graph and $\tilde c \in \mathbb{Z}_+^E$ be a stochastic edge weight.
The bipartite matching problem can then be represented as
\begin{align}\label{eq:bipartite}
\begin{array}{lll}
\text{maximize} &\ \displaystyle\sum_{e \in E} \tilde c_e x_e \\
\text{subject to} &\ \displaystyle\sum_{e \in \delta(u)} x_e \le 1 &\quad (u \in V), \\
&\ x \in \{0,1\}^E,
\end{array}
\end{align}
where $\delta(u) = \{\, e \in E : u \in e \,\}$.
The K\H{o}nig--Egerv\'ary theorem~\cite{kHonig1931graphs,egervary1931combinatorial} shows that the LP relaxation of this system is TDI, so Algorithm~\ref{alg:adaptive} has an approximation ratio of $(1 - \epsilon)$ with high probability for sufficiently large $T$.
Moreover, if the algorithm finds an integral solution to the optimistic problem in Line~\ref{line:2}, it reveals a matching in each iteration, and hence at most $T$ edges per vertex in total.
Finally, by Lemma~\ref{lem:tdi}, there exists an $(\epsilon, \epsilon)$-witness cover of size $e^{O(\mu \log n)}$ for each $\mu \geq 1$. 
We can therefore obtain the following result from Theorem~\ref{thm:adaptive}.
\begin{corollary}\label{cor:bipartite}
By taking $T = \Omega(\Delta_c \log (n/\epsilon)/\epsilon p)$,
for the bipartite matching problem,
Algorithm~\ref{alg:adaptive} outputs a $(1 - \epsilon)$-approximate solution with probability at least $1 - \epsilon$.
\end{corollary}
This result is improved by using the vertex sparsification lemma in Section~\ref{sec:sparsification}.

\subsubsection{Non-bipartite Matching}\label{sec:nonbipartite}

Next, we consider the non-bipartite matching problem.
Let $(V, E)$ be a graph and $\tilde c \in \mathbb{Z}_+^E$ be a stochastic edge weight.
A naive formulation of this problem is as follows:
\begin{align}
\begin{array}{lll}
\text{maximize} &\ \displaystyle\sum_{e \in E} \tilde c_e x_e \\
\text{subject to} &\ \displaystyle\sum_{e \in \delta(u)} x_e \le 1 &\quad (u \in V), \\
&\ x \in \{0,1\}^E.
\end{array}
\end{align}
It is known that the LP relaxation of this system is TDI/2 (totally dual half-integral) and has an integrality gap of $3/2$~\cite{schrijver2003combinatorial}.
Therefore, by the same argument as in the bipartite matching problem, we can show that Algorithm~\ref{alg:adaptive} has an approximation ratio of $(2 - \epsilon)/3$ with high probability if $T = \Omega(\Delta_c \log (n/\epsilon)/\epsilon p)$.

To improve this approximation ratio, we consider a strengthened formulation by adding the blossom inequalities:
\begin{align}\label{eq:blossom}
\begin{array}{lll}
\text{maximize} &\ \displaystyle\sum_{e \in E} \tilde c_e x_e \\
\text{subject to} &\ \displaystyle\sum_{e \in \delta(u)} x_e \le 1 &\quad (u \in V), \\
&\ \displaystyle\sum_{e \in E(S)} x_e \le \floor{\frac{|S|}{2}} &\quad (S \in \mathcal{V}_{\rm odd}), \\
&\ x \in \{0,1\}^E,
\end{array}
\end{align}
where $E(S) = \{\, e \in E : e \subseteq S \,\}$ and $\mathcal{V}_{\rm odd} = \{\, S \subseteq V : |S| \textrm{ is odd and at least  } 3 \,\}$.
Cunningham and Marsh~\cite{cunningham1978primal} showed that this system is TDI, so our algorithm has an approximation ratio of $(1 - \epsilon)$ with high probability for sufficiently large $T$.
Moreover, the number of revealed edges is at most $T$ per vertex. %since it reveals a matching in each iteration.

The only remaining issue is the number of iterations.
Since the system \eqref{eq:blossom} has exponentially many constraints, we have to exploit its combinatorial structure to reduce the number of possibilities. 
The dual problem is given by
\begin{align}
\label{eq:matchingdual}
\begin{array}{lll}
\text{minimize} &\ \displaystyle\sum_{u \in V} y_u + \sum_{S \in \mathcal{V}_{\rm odd}} \floor{\frac{|S|}{2}} z_S =: \tau(y, z)\\
\text{subject to} &\ \displaystyle y_u + y_v + \sum_{S \in \mathcal{V}_{\rm odd} \colon \{u,v\} \subseteq S} z_S \ge \tilde c_{e} &\quad (e = \{u,v\} \in E), \\
&\ y \in \mathbb{R}_+^V,\ z \in \mathbb{R}_+^{\mathcal{V}_{\rm odd}}.
\end{array}
\end{align}
For $\mu \geq 1$,
let us define a set $W \subseteq \mathbb{R}_+^V \times \mathbb{R}_+^{\mathcal{V}_{\rm odd}}$ by
\begin{align}
  W = \left\{\, (y,z) \in \mathbb{Z}_+^V \times \mathbb{Z}_+^{\mathcal{V}_{\rm odd}} :\, \tau(y, z) \le (1 - \epsilon) \mu \,\right\}.
\end{align}
It is clear that $W$ is an $(\epsilon, \epsilon)$-witness cover, and we can evaluate the size of $W$ as follows.

\begin{Claim} \label{claim:nonbipartite}
$|W| = e^{O(\mu \log n)}$.
\end{Claim}

\begin{proof}
%It is clear that $W$ is an $(\epsilon, \epsilon)$-witness cover, and we can evaluate the size of $W$, by counting $y$ and $z$ separately.
We count the candidates for $y$ and $z$ separately.

Since $\sum_i y_i < \mu$, the number of candidates for $y$ is at most
\begin{align}
	\sum_{k=0}^{\floor{\mu}} \binom{n}{k} \le \left(\floor{\mu} + 1\right) \left( \frac{e n}{\floor{\mu}} \right)^{\floor{\mu}} = e^{O(\mu \log n)}.
\end{align}

To count the number of candidates for $z$, we regard $z$ as a multiset, e.g., if $z_S = 2$ then we think there are two $S$.
Let $s_i$ $(i = 1, 2, \ldots, k)$ be the size of each set contained in $z$. 
Then we have
\begin{align}
	s_1 + \cdots + s_k \le 3 \left( \floor{\frac{s_1}{2}} + \cdots + \floor{\frac{s_k}{2}} \right) < 3 \mu.
\end{align}
Therefore, the number of candidates for $z$ is given by
\begin{align}
	\sum_{k=0}^{\floor{\mu}} \sum_{\substack{s_1 + \cdots + s_k \le 3\mu \\ s_1, \ldots, s_k \ge 3}} \binom{n}{s_1} \cdots \binom{n}{s_k}
        \le \sum_k \binom{3 \floor{\mu} + k -1}{k} n^{\mu} = e^{O(\mu \log n)}.
\end{align}
By multiplying the number of candidates for $y$ and $z$, we obtain the required result.
\end{proof}
Therefore, we obtain the following.
\begin{corollary}
By taking $T = \Omega(\Delta_c \log (n/\epsilon)/\epsilon p)$,
for the non-bipartite matching problem,
Algorithm~\ref{alg:adaptive} outputs a $(1 - \epsilon)$-approximate solution with probability at least $1 - \epsilon$.
\end{corollary}

The same analysis can be applied to the (simple) $b$-matching problem.

\paragraph{Relationship to the analysis of Assadi et al.}
\label{sec:matching}

The adaptive and non-adaptive algorithms of Assadi et al.~\cite{assadi2016stochastic}
for the unweighted non-bipartite matching problem are within our framework (Algorithms~\ref{alg:adaptive} and \ref{alg:nonadaptive}, respectively), 
and their analysis utilizes the Tutte--Berge formula.
They showed that the required number of iterations is $O(\log (n/\epsilon \mu)/\epsilon p)$, and it is reduced to $O(\mathrm{poly}(p, 1/\epsilon))$ by using the vertex sparsification lemma.

For the unweighted problem, our analysis gives a weaker result than theirs.
However, since no simple alternative to the Tutte--Berge formula for the weighted problem is known, our analysis is more general than theirs.

\subsubsection{$k$-Hypergraph Matching}

Let $(V, E)$ be a $k$-uniform hypergraph, i.e., $E$ is a set family on $V$ whose each element $e \in E$ has size exactly $k$.
Let $\tilde c \in \mathbb{Z}_+^E$ be a stochastic edge weight.
The \emph{$k$-hypergraph matching problem} can be represented as
\begin{align}\label{eq:k-hypergraph}
\begin{array}{lll}
\text{maximize} &\ \displaystyle\sum_{e \in E} \tilde c_e x_e \\
\text{subject to} &\ \displaystyle\sum_{e \in \delta(u)} x_e \le 1 &\quad (u \in V), \\
&\ x \in \{0,1\}^E,
\end{array}
\end{align}
where $\delta(u) = \{\, e \in E : u \in e \,\}$.
Chan and Lau~\cite{chan2010linear} proved that the LP relaxation of the above system has an integrality gap of $\alpha := 1/(k - 1 + 1/k)$, and they also proposed an LP-relative $\alpha$-approximation algorithm. %for the non-stochastic case.
Since at most one edge per vertex is revealed in expectation due to the constraint $\sum_{e \in \delta(u)} x_e \le 1$,
the expected number of revealed hyperedges per vertex is $O(T)$ in total.

The only remaining issue is the number of iterations.
Since the system \eqref{eq:k-hypergraph} has polynomially many constraints and is not TDI, we have to discretize the dual variables.
The corresponding dual problem is given by
\begin{align}\label{eq:dual_k-hypergraph}
\begin{array}{lll}
\text{minimize} &\ \displaystyle\sum_{u \in V} y_u \\
\text{subject to} &\ \displaystyle\sum_{u \in e} y_u \ge \tilde c_e &\quad (e \in E), \\ 
&\ y \in \mathbb{R}_+^V.
\end{array}
\end{align}
Here, we show that the dual optimal solution is sparse.
\begin{Claim}
If the optimal value is less than $\mu$,
then there exists a dual optimal solution $y \in \mathbb{R}_+^V$ such that $|\supp(y)| < k \mu$.
\end{Claim}
\begin{proof}
Let $x \in \mathbb{R}_+^E$ be a primal optimal solution. 
We can assume that $x_e > 0$ only if $\tilde c_e \ge 1$.
Therefore, we have $\sum_e x_e \le \sum_e \tilde c_e x_e < \mu$. 
Since each hyperedge consists of exactly $k$ elements, we have
\begin{align}
	\sum_{u \in V} \sum_{e \in \delta(u)} x_e = k \sum_{e \in E} x_e < k \mu.
\end{align}
This shows that less than $k \mu$ inequalities can hold in equality.
Therefore, by complementary slackness, the corresponding dual optimal solution $y$ satisfies $|\supp(y)| < k \mu$.
\end{proof}
Therefore, by Lemma~\ref{lem:poly_nontdi_sparse},
there exists an $(\epsilon, \epsilon/2)$-witness cover for each $\mu \geq 1$
of size at most $M^\mu$, where $M = \exp\left(O\left(k \log \frac{n}{k\mu} + \frac{1}{\epsilon}\right)\right)$.
Thus we obtain the following.
\begin{corollary}
\label{cor:khypergraphmatching1}
By taking $T = \Omega(\Delta_c(k \log (n/\epsilon) + 1/\epsilon)/\epsilon p)$, for the $k$-hypergraph matching problem, Algorithm~\ref{alg:adaptive} outputs a $(1 - \epsilon)/(k - 1 + 1/k)$-approximate solution with probability at least $1 - \epsilon$.
\end{corollary}
This result is improved by using the vertex sparsification lemma in Section~\ref{sec:sparsification}.

\paragraph{Comparison with Blum et al.}

Blum et al.~\cite{blum2015ignorance} provided adaptive and non-adaptive algorithms for the unweighted $k$-hypergraph matching problem based on local search of Hurkens and Schrijver~\cite{hurkens1989size}.
Their adaptive algorithm has approximation ratio of $(2 - \epsilon)/k$ in expectation by conducting a constant number of queries per vertex.

For unweighted problem, our algorithm has a worse approximation ratio than theirs.
However, our algorithm has four advantages:
it requires exponentially smaller number of queries;
it runs in polynomial time both in $n$ and $1/\epsilon$;
it is applied to the weighted problem with the same approximation ratio;
and it has a stronger stochastic guarantee, i.e., not in expectation but with high probability.

\begin{remark}
For unweighted $k$-hypergraph matching problem, Chan and Liu~\cite{chan2010linear} showed that there is a packing LP with an integrality gap of $2/(k+1)$. 
Note that the rounding algorithm for this LP is not known.
Using this formulation, we obtain a $(2 - \epsilon)/(k+1)$ approximation algorithm which conducts $O_{\epsilon,p}(\log^2 n)$ queries and runs in non-polynomial time (i.e., it performs exhaustive search).
\end{remark}

\subsubsection{$k$-Column Sparse Packing Integer Programming}\label{sec:k-CSPIP}

The $k$-column sparse packing integer programming problem is a common generalization of the $k$-hypergraph matching problem and the knapsack problem, and can be represented as follows
(the formulation itself just rewrites \eqref{eq:packing} by using the entries of the matrix and of the vectors):
\begin{align}\label{eq:k-column_sparse}
\begin{array}{lll}
\text{maximize} &\ \displaystyle\sum_{j = 1}^m \tilde c_j x_j \\
\text{subject to} &\ \displaystyle\sum_{j = 1}^m a_{ij} x_j \le b_i &\quad (i = 1, \ldots, n), \\
&\ x \in \{0,1\}^m,
\end{array}
\end{align}
where ``$k$-column sparse'' means that $| \{\, i  : a_{ij} \neq 0 \,\} | \le k$ for each $j \in \{1, \ldots, m\}$.
Without loss of generality, we assume that $a_{ij} \le b_i$ for all $i \in \{1,\ldots,n\}$ and $j \in \{1,\ldots,m\}$.

The main difference from the other problems is that the system $\sum_j a_{ij} x_j \le b_i$, $x \geq 0$ does not imply $x_j \le 1$. 
Instead, we have $x_j \le w_j$, where $w_j = \min_{i \colon a_{ij} \neq 0} b_i/a_{ij}$.
Let $w = \max_j w_j$.
By modifying Algorithm~\ref{alg:adaptive} to reveal each $x_j$ with probability $x_j / w$, we obtain the same approximation guarantee with $w$ times larger number of iterations.

Parekh~\cite{parekh2011iterative} proposed an LP-relative $(1/2k)$-approximation algorithm for general $k$ and a $(1/3)$-approximation algorithm for $k = 2$, which encompasses the \emph{demand matching problem}~\cite{shepherd2007demand}. 
The expected number of revealed elements for each constraint $i$ is $O(b_iT)$, because we have
\begin{align}
&\sum_{j\colon a_{ij} \neq 0} x_j \le \sum_j a_{ij} x_j \le b_i. %,\quad \text{and}\\
%&\sum_{j\colon a_{ij} \neq 0} x_j \le \sum_{j\colon a_{ij} \neq 0} w \le kw.
\end{align}
%\COMM{YY}{One of the two types of bounds disappears.}

The only remaining issue is the number of iterations.
Using the same approach as for the $k$-hypergraph matching problem, we obtain the following result.
\begin{corollary}
By taking $T = \Omega(\Delta_c w (k \log(n/\epsilon) + 1/\epsilon)/\epsilon p)$,
for the $k$-column sparse packing integer programing problem, Algorithm~\ref{alg:adaptive} outputs a $(1 - \epsilon)/2k$-approximate solution with probability at least $1 - \epsilon$.
\end{corollary}

\subsection{Matroid Problems}
Now we apply our technique to matroid-related optimization problems
(see, e.g., \cite{schrijver2003combinatorial} for basics of matroids and related optimization problems).
In this section, unless otherwise noted,
$m$ denotes the ground set size of the matroids in question.

\subsubsection{Maximum Independent Set}

Let ${\mathbf M} = (E, \mathcal{I})$ be a matroid on a finite set $E$, and let $r\colon 2^E \to \mathbb{Z}_+$ be its rank function.
A set $S \subseteq E$ is a \emph{flat} if $r(S \cup e) \ne r(S)$ for all $e \in E \setminus S$,
and let ${\mathcal F}_{\mathbf M}$ denote the family of flats in ${\mathbf M}$.
For a subset $S \subseteq E$,
the smallest flat containing $S$ is called the \emph{closure} of $S$.
We assume that the rank of the matroid is relatively small to ensure that Algorithm~\ref{alg:adaptive} does not reveal all the elements.

Let $\tilde c \in \mathbb{Z}_+^E$ be a stochastic weight.
The maximum independent set problem can be represented as follows:
\begin{align}
\begin{array}{lll}
\text{maximize} &\ \displaystyle\sum_{e \in E} \tilde c_e x_e \\
\text{subject to} &\ \displaystyle\sum_{e \in S} x_e \le r(S) &\quad (S \subseteq E), \\
&\ x \in \{0,1\}^E.
\end{array}
\end{align}
Edmonds~\cite{edmonds1970submodular} showed that the LP relaxation of the above system is TDI, so our algorithm has an approximation ratio of $(1 - \epsilon)$ with high probability.
Moreover, the number of revealed elements is $O(r T)$, where $r = r(E)$ is the rank of the matroid in question.

The only remaining issue is the number of iterations.
Since the system has exponentially many constraints, we have to exploit the combinatorial structure of the problem to reduce the number of possibilities.
The dual problem is given by
\begin{align}
\begin{array}{lll}
\text{minimize} &\ \displaystyle\sum_{S \subseteq E} r(S) y_S \\
\text{subject to} &\ \displaystyle\sum_{S \subseteq E\colon e \in S} y_S \ge \tilde c_e &\quad (e \in E), \\
&\ y \in \mathbb{R}_+^{2^E}.
\end{array}
\end{align}
Since the closure of each $S \subseteq E$ contributes the objective value by $r(S)$ and contains all elements in $S$,
we can restrict the supports of $y$ to the subfamilies of ${\mathcal F}_{\mathbf M}$.
Then, for $\mu \geq 1$,
let us define a set $W \subseteq \mathbb{R}_+^{2^E}$ by
\begin{align}
	W = \left\{\, y \in \mathbb{Z}_+^{2^E} :\, \sum_{S \in {\mathcal F}_{\mathbf M}} r(S) y_S \le (1 - \epsilon) \mu,\ \supp(y) \subseteq {\mathcal F}_{\mathbf M} \, \right\}.
\end{align}
It is clear that $W$ is an $(\epsilon, \epsilon)$-witness cover, and we can evaluate the size of $W$ as follows.

\begin{Claim}\label{cl:maximum_independent_set}
$|W| = e^{O(\mu \log m)}$.
\end{Claim}
\begin{proof}
To evaluate the size of $W$, as for the non-bipartite matching problem, we regard $y$ as a multiset of flats. 
Let $r_1, \ldots, r_k$ be the ranks of flats in $y$. 
Then we have $r_1 + \cdots + r_k < \mu$.
Since each flat is the closure of some independent set, the number of flats of rank $r$ is at most the number of independent sets of size $r$, which is at most $\binom{m}{r}$.
Therefore, the number of dual candidates for $y$ is given by
\begin{align}
	\quad \sum_{k=0}^{\floor{\mu}} \sum_{\substack{r_1 + \cdots + r_k \leq \mu \\ r_1, \ldots, r_k \ge 1}} \binom{m}{r_1} \cdots \binom{m}{r_k} \le e^{O(\mu \log m)}. \quad \qedhere
\end{align}
\end{proof}
Therefore, we obtain the following.
\begin{corollary}
By taking $T = \Omega(\Delta_c \log (m / \epsilon)/\epsilon p)$,
for the maximum independent set problem,
Algorithm~\ref{alg:adaptive} outputs a $(1 - \epsilon)$-approximate solution with probability at least $1 - \epsilon$.
\end{corollary}

\begin{comment}
\begin{remark}
By conducting a different analysis, we can take $T = \Omega(1/\epsilon p)$. See Appendix.
\COMM{TM}{APPENDIX}
\end{remark}
\end{comment}

\subsubsection{Matroid Intersection}

The same technique can also be applied to the matroid intersection problem.
Let ${\mathbf M}_j = (E, \mathcal{I}_j)$ $(j = 1, 2)$ be two matroids whose rank functions are $r_j \colon 2^E \to \mathbb{Z}_+$, and $\tilde c \in \mathbb{Z}_+^E$ be a stochastic weight.
The weighted matroid intersection problem can be represented as 
\begin{align}
\begin{array}{lll}
\text{maximize} &\ \displaystyle\sum_{e \in E} \tilde c_e x_e \\
\text{subject to} &\ \displaystyle\sum_{e \in S} x_e \le r_j(S) &\quad (S \subseteq E, \ j \in \{1, 2\}), \\
&\ x \in \{0,1\}^{E}.
\end{array}
\end{align}
Edmonds~\cite{edmonds1970submodular} showed that the LP relaxation of the above system is TDI, so our algorithm has an approximation ratio of $(1 - \epsilon)$ with high probability for sufficiently large $T$.
Moreover, the number of revealed elements is $O(r^\ast T)$,
where $r^\ast$ is the maximum rank of a common independent set in the two matroids.

The only remaining issue is the number of iterations.
As the analysis of the maximum independent set problem,
we can restrict the supports of dual vectors $y \in {\mathbb R}_+^{2^E} \times {\mathbb R}_+^{2^E}$ to the subfamilies of ${\mathcal F}_{{\mathbf M}_1} \times {\mathcal F}_{{\mathbf M}_2}$,
and we obtain the following by the same argument.
\begin{corollary}
By taking $T = \Omega(\Delta_c \log (m/\epsilon)/\epsilon p)$, 
for the matroid intersection problem, Algorithm~\ref{alg:adaptive} outputs a $(1 - \epsilon)$-approximate solution with probability at least $1 - \epsilon$.%
\footnote{Note that the bipartite matching problem is a special case of the matroid intersection problem, and Corollary~\ref{cor:bipartite} is obtained from a naive application of this result. Using the vertex sparsification lemma shown in Section~\ref{sec:sparsification}, a stronger result is obtained for bipartite matching (Corollary~\ref{cor:bipartite_stronger}).}
\end{corollary}

\subsubsection{$k$-Matroid Intersection}

Let ${\mathbf M}_j = (E, \mathcal{I}_j)$ $(j = 1, 2, \ldots, k)$ be $k$ matroids whose rank functions are $r_j \colon 2^E \to \mathbb{Z}_+$, and $\tilde c \in \mathbb{Z}_+^E$ be a stochastic weight.
The $k$-matroid intersection problem can be represented as 
\begin{align}\label{eq:k-matroid_intersection}
\begin{array}{lll}
\text{maximize} &\ \displaystyle\sum_{e \in E} \tilde c_e x_e \\
\text{subject to} &\ \displaystyle\sum_{e \in S} x_e \le r_j(S) &\quad (S \subseteq E,\ j \in \{1, 2, \ldots, k\}), \\
&\ x \in \{0,1\}^E.
\end{array}
\end{align}
The important difference between the $2$-intersection and $k$-intersection ($k \ge 3$) problems is that the latter is NP-hard in the non-stochastic case.
Moreover, the LP relaxation of the system is not a kind of TDI.

Adamczyk et al.~\cite{adamczyk2016submodular} proposed an LP-relative $(1/k)$-approximation algorithm. 
The expected number of revealed elements is $O(\hat{r} T)$, where $\hat{r}$ is the minimum of the ranks of the $k$ matroids.

The only remaining issue is the number of iterations.
Since the LP relaxation of \eqref{eq:k-matroid_intersection} is not TDI, we have to discretize the dual variables.
Moreover, since we could not prove the dual optimal solution is sparse, we use Theorem~\ref{thm:kolliopoulos2005}.

For $\theta = \Theta\left(\frac{\epsilon^2}{\log m}\right)$ and $\mu \geq 1$, %be defined in Theorem~\ref{thm:kolliopoulos2005},
let us define a set $W \subseteq \left(\mathbb{R}_+^{2^V}\right)^k$ by
\begin{align}
  W = \biggl\{\, &(y^1, \ldots, y^k) \in \left(\theta\mathbb{Z}_+^{2^V}\right)^k \, : \notag\\& \sum_{j,S} r_j(S) y^j_S \le \left(1 - \frac{\epsilon}{2}\right) \mu,\ \supp(y^j) \subseteq {\mathcal F}_{{\mathbf M}_j} \ (\forall j = 1, \ldots, k) \,\biggr\}.
\end{align}
By Theorem~\ref{thm:kolliopoulos2005}, $W$ is an $(\epsilon, \epsilon/2)$-witness cover,
and we can evaluate its size as follows.

\begin{Claim}
$|W| = e^{O(\mu k \log^2 m/\epsilon^2)}$.
\end{Claim}
\begin{proof}
To evaluate the size of $W$, as for the maximum independent set problem, we count each $y^j$ separately,
%As same as the maximum independent set case,
where we regard $y^j$ as a multiset in which each flat $S$ contributes $\theta$. 
Let $r_1, \ldots, r_k$ be the ranks of flats in $y^j$. 
Then we have $r_1 + \cdots + r_k < \mu / \theta$.
By the same argument as for the maximum independent set problem, the number of dual candidates for $y^j$ is at most $e^{O((\mu/\theta)\log m)}$.
By multiplying the numbers of candidates for the $k$ coordinates, we obtain the required result
(recall $\theta = \Theta(\epsilon^2/\log m)$).
\end{proof}
Therefore, we obtain the following.
\begin{corollary}
By taking $T = \Omega(\Delta_c k \log m \log (m/\epsilon)/\epsilon^3 p)$,
for the $k$-matroid intersection problem, 
Algorithm~\ref{alg:adaptive} outputs a $(1 - \epsilon)/k$-approximate solution with probability at least $1 - \epsilon$. 
\end{corollary}

\subsubsection{Matchoid}
The matchoid problem is a common generalization of the matching problem and the matroid intersection problem.
Let $(V, E)$ be a graph with $|V| = n$ and $|E| = m$,
${\mathbf M}_v = (\delta(v), \mathcal{I}_v)$ be a matroid whose rank function is $r_v \colon 2^{\delta(v)} \to \mathbb{Z}_+$ for each vertex $v \in V$,
and $\tilde c \in \mathbb{Z}_+^E$ be a stochastic edge weight.
The task is to find a maximum-weight subset $F \subseteq E$ of edges such that $F \cap \delta(v) \in \mathcal{I}_v$ for every $v \in V$.
A naive LP formulation is as follows:
\begin{align}
\begin{array}{lll}
\text{maximize} &\ \displaystyle\sum_{e \in E} \tilde c_e x_e \\
\text{subject to} &\ \displaystyle\sum_{e \in S} x_e \le r_v(S) &\quad (v \in V,~S \subseteq \delta(v)), \\
&\ x \in \{0, 1\}^E.
\end{array}
\end{align}
Lee, Sviridenko, and Vondr\'ak \cite{lee2013matroid}\footnote{Precisely,
the discussion is given via a reduction to the matroid matching problem, which preserves the variables and the feasible region.}
proposed an LP-relative $(2/3)$-approximation algorithm.
For each vertex $v \in V$,
the expected number of revealed edges incident to $v$ is $O(r_vT)$, where $r_v = r_v(\delta(v)) \ge \sum_{e \in \delta(v)} x_e$.

The only remaining issue is the number of iterations.
Since it has exponentially many constraints (in the maximum degree), we have to exploit its combinatorial structure to reduce the number of possibilities.
The dual problem is given by
\begin{align}
\begin{array}{lll}
\text{minimize} &\ \displaystyle\sum_{v \in V}\sum_{S \subseteq \delta(v)} r_v(S) y_{v, S} \\
\text{subject to} &\ \displaystyle\sum_{v \in V}\sum_{S \subseteq \delta(v)\colon e \in S} y_{v, S} \ge \tilde c_e &\quad (e \in E), \\
&\ y \in \mathbb{R}_+^{\mathcal{S}},
\end{array}
\end{align}
where $\mathcal{S} = \{\, (v, S) \mid v \in V,~S \subseteq \delta(v) \,\} \subseteq V \times 2^E$.
Similarly to the other matroid problems, we can restrict the support of $y$ so that, if $y_{v, S} > 0$, then $S \subseteq \delta(v)$ is a flat in $\mathbf{M}_v$.
Let $\mathcal{F}_v$ be the set of flats in $\mathbf{M}_v$ and $\mathcal{F} = \{\, (v, S) \mid v \in V,~S \in \mathcal{F}_v \,\}$.
Based on the (TDI/2)-ness of the matroid matching polyhedron due to Gijswijt and Pap~\cite{gijswijt2013algorithm},
the following set is an $(\epsilon, \epsilon)$-witness cover:
\begin{align}
	W = \left\{\, y \in \frac{1}{2} \mathbb{Z}_+^{\mathcal{S}} : \, \sum_{v \in V}\sum_{S \subseteq \delta(v)} r_v(S) y_{v, S} \le (1 - \epsilon) \mu,\ \supp(y) \subseteq {\mathcal F} \,\right\}.
\end{align}
Similarly to the maximum independent set case, $|W|$ can be bounded by $e^{O(\mu \log m)}$.
Thus we obtain the following.
\begin{corollary}
By taking $T = \Omega(\Delta_c \log (m/\epsilon)/\epsilon p)$, for the matchoid problem, \mbox{Algorithm~\ref{alg:adaptive}} outputs a $(2 - \epsilon)/3$-approximate solution with probability at least $1 - \epsilon$.
\end{corollary}

\subsubsection{Degree Bounded Matroid}

Let $(V, E)$ be a hypergraph with $|V| = n$ whose maximum degree is $d = \max_{u \in V} |\delta(u)|$, and $b \colon E \to \mathbb{Z}_+$ give a capacity of each hyperedge.
Let ${\mathbf M} = (V, \mathcal{I})$ be a matroid whose rank function is $r \colon 2^V \to \mathbb{Z}_+$,
and $\tilde c \in \mathbb{Z}_+^V$ be a stochastic weight.
The degree bounded matroid problem can be represented as 
\begin{align}
\begin{array}{lll}
\text{maximize} &\ \displaystyle\sum_{u \in V} \tilde c_u x_u \\
\text{subject to} &\ \displaystyle\sum_{u \in e} x_u \le b(e) &\quad (e \in E), \\ 
&\ \displaystyle\sum_{u \in S} x_u \le r(S) &\quad (S \subseteq V), \\
&\ x \in \{0,1\}^V.
\end{array}
\end{align}
Kir{\'a}ly et al.~\cite{kiraly2012degree} proposed an %LP-relative approximation
algorithm that finds a (possibly infeasible) solution whose objective value is at least the LP-optimal value
and which violates each capacity constraint by at most $d - 1$. %~\cite{kiraly2012degree}.
Since $\sum_{u \in V} x_u \le r(V) =: r$,
the expected number of revealed elements is $O(r T)$.

The only remaining issue is the number of iterations.
Since this system has exponentially many constraints, we have to exploit its combinatorial structure to reduce the number of possibilities.
The dual problem is given by
\begin{align}\label{eq:dual_DBM}
\begin{array}{lll}
\text{minimize} &\ \displaystyle\sum_{e \in E} b(e) y_e + \sum_{S \subseteq V} r(S) z_S =: \tau(y, z)\\
\text{subject to} &\ \displaystyle\sum_{e \in \delta(u)} y_e + \sum_{S \subseteq V \colon u \in S} z_S \ge \tilde c_u &\quad (u \in V), \\
&\ y \in \mathbb{R}_+^E,\ z \in \mathbb{R}_+^{2^V}.
\end{array}
\end{align}
Since the system is not TDI, we have to discretize the dual variables.
We could not prove the sparsity of $z$ but, by observing the sparsity of $y$ and using the matroid property, we can see that there exists a good discretization.

\begin{Claim}
Let $(y,z) \in \mathbb{R}_+^E \times \mathbb{R}_+^{2^V}$ be an optimal solution to \eqref{eq:dual_DBM} with $\tau(y, z) < \mu$.
Then, there exists a feasible solution $(y',z')$ with $\tau(y', z') < (1 - \epsilon/2)\mu$
whose entries are multiple of $\epsilon/2d$.
%and its objective value is less than $(1 - \epsilon/2) \mu$.
\end{Claim}
\begin{proof}
Let $x \in \mathbb{R}_+^V$ be a primal LP-optimal solution.
By complementary slackness, $y_e > 0$ only if the constraint $\sum_{u \in e} x_u \le b(e)$ holds in equality. 
Thus, by summing up, we have 
\begin{align}
	\sum_{e\colon y_e > 0} b(e) = \sum_{e\colon y_e > 0} \sum_{u \in e} x_u \le \sum_{e \in E} \sum_{u \in e} x_u \le d \sum_{u \in V} x_u < d \mu.
\end{align}

Now we round up each entry of $y$ to the minimum multiple of $\epsilon/2 d$ to obtain $y'$.
This increases objective value at most $(\epsilon/2d) \sum_{e\colon y_e > 0} b(e) < \epsilon \mu/2$.
Therefore the objective value of $(y',z)$ is at most $(1 - \epsilon/2) \mu$.
To discretize $z$, we consider the minimization problem with respect to $z$:
\begin{align}
\begin{array}{lll}
\text{minimize} &\ \displaystyle\sum_{S \subseteq V} r(S) z_S \\
\text{subject to} &\ \displaystyle\sum_{S \subseteq V \colon u \in S} z_S \ge \tilde c_u - \sum_{e \in \delta(u)} y'_e &\quad (u \in V).
\end{array}
\end{align}
This problem is the dual of the maximum independent set problem whose cost vector is a multiple of $\epsilon/2 d$.
Therefore, by the TDIness of the maximum independent set problem,
there exists an optimal solution $z'$ whose entries are multiples of $\epsilon/2 d$.
Thus, $(y',z')$ is feasible by construction and has an objective value of at most $(1 - \epsilon/2) \mu$.
\end{proof}

As for the other matroid problems, we can assume that $\supp(z)$ is a set of flats.
For $\mu \geq 1$,
let us define a set $W \subseteq \mathbb{R}_+^E \times \mathbb{R}_+^{2^V}$ by
\begin{align}
  W = \left\{\, (y,z) \in \frac{\epsilon}{2d}\left(\mathbb{Z}_+^E \times \mathbb{Z}_+^{2^V}\right) \, : \ \tau(y, z) \leq \left(1 - \frac{\epsilon}{2}\right) \mu,\ \supp(z) \subseteq {\mathcal F}_{\mathbf M} \,\right\}.
\end{align}
By construction, $W$ is an $(\epsilon,\epsilon/2)$-witness cover,
and we can evaluate its size as follows.

\begin{Claim}
$|W| = e^{O\left(\mu d \log n / \epsilon\right)}$.
\end{Claim}
\begin{proof}
To evaluate the size of $W$, we separately count $y$ and $z$.
%The number of candidates for $y$ is $e^{O(d \mu (\log (m/\mu) + 1/\epsilon))}$, which is counted similarly to Lemma~\ref{lem:poly_nontdi_sparse},
%where $m = |E| = O(dn)$.
The number of candidates for $y$ is evaluated as similar to Lemma~\ref{lem:tdi} by distributing $2 d \mu/\epsilon$ tokens to $|E| = O(dn) = O(n^2)$ components, and is bounded by $e^{O(\mu d \log n / \epsilon)}$.
The number of candidates for $z$ is evaluated as similar to Claim~\ref{cl:maximum_independent_set} by counting multisets with weight $\epsilon/2 d$, and is bounded by $e^{O(\mu d \log n / \epsilon)}$.
By multiplying these two numbers of candidates, we obtain the required result.
%\COMM{YY}{The proof is difficult to follow. (Correct?)}
\end{proof}
Therefore, we obtain the following.
\begin{corollary}
For the degree bounded matroid problem with maximum degree $d$, by taking $T = \Omega(\Delta_c d \log (n/\epsilon) / \epsilon^2 p)$, Algorithm~\ref{alg:adaptive} outputs a $(1 - \epsilon)$-approximate solution that violates each constraint at most $d - 1$ with probability at least $1 - \epsilon$. 
\end{corollary}

\subsection{Stable Set Problems}
We finally show applications to stable set problems.
In this section,
$n$ and $m$ denote the numbers of vertices and edges, repsectively, in the graph in question.

\subsubsection{Stable Set in Some Perfect Graphs}

We assume that the stability number $\alpha$ (the maximum size of a stable set) is relatively small to ensure that Algorithm~\ref{alg:adaptive} does not reveal all the vertices.
By the Tur\`an theorem~\cite{turan1954theory}, the average degree is required to be relatively large.

Let $(V, E)$ be a graph and $\tilde c \colon V \to \mathbb{Z}_+$ be a stochastic vertex weight.
The maximum stable set problem can be represented as
\begin{align}
\begin{array}{lll}
\text{maximize} &\ \displaystyle\sum_{u \in V} \tilde c_u x_u \\
\text{subject to} &\ x_u + x_v \le 1 &\quad ((u, v) \in E), \\
&\ x \in \{0,1\}^V.
\end{array}
\end{align}
The LP relaxation of this system is half-integral.
However, this is not helpful because the number of revealed vertices can be large: there is a solution $x_u = 1/2$ for all $u \in V$, which corresponds to revealing half of the vertices in expectation.

We instead consider the following formulation, which introduces the \emph{clique inequalities}:
\begin{align}
\begin{array}{lll}
\text{maximize} &\ \displaystyle\sum_{u \in V} c_u x_u \\
\text{subject to} &\ \displaystyle\sum_{u \in C} x_u \le 1 &\quad (C \in \mathcal{C}), \\
&\ x_u \in \{0,1\}^V,
\end{array}
\end{align}
where $\mathcal{C}$ is the set of maximal cliques.
A graph is \emph{perfect} if the LP relaxation of the above system is TDI.
If we assume that the graph is perfect, Algorithm~\ref{alg:adaptive} has an approximation ratio of $(1 - \epsilon)$ with high probability for sufficiently large $T$, and the number of revealed vertices is $O(\alpha T)$ in expectation.

The dual problem is given by
\begin{align}
\begin{array}{lll}
\text{minimize} &\ \displaystyle\sum_{C \in \mathcal{C}} y_C \\
\text{subject to} &\ \displaystyle\sum_{C \in \mathcal{C}\colon u \in C} y_C \ge c_u &\quad (u \in V), \\
&\ y \in \mathbb{R}_+^{\mathcal{C}}.
\end{array}
\end{align}
If the number of maximal cliques is $O(n^k)$ for some fixed constant $k$, we immediately see that the required number of iterations is $O(k \log (n/\epsilon)/\epsilon p)$.
A perfect graph may have exponentially many maximal cliques in general, but the following graph classes have at most polynomially many maximal cliques.

\begin{itemize}
  \item 
    If a graph is chordal, it has only linear number of maximal cliques.

  \item
    If a graph has a bounded clique number (i.e., the size of cliques are bounded by a constant $k$), the number of cliques is at most $\binom{n}{k} = O(n^k)$.
    This includes a graph class that can be characterized by forbidden minors and subgraphs.
\end{itemize}

\subsubsection{Stable Set in t-Perfect Graphs}

Another tractable graph class for the stable set problem is t-perfect graphs.
A graph $(V, E)$ is \emph{$t$-perfect} if the relaxation of the following formulation is integral, i.e., it has an integral optimal solution:
\begin{align}
\label{eq:tperfect}
\begin{array}{lll}
\text{maximize} &\ \displaystyle\sum_{u \in V} \tilde c_u x_u \\
\text{subject to} &\ x_u + x_v \le 1 &\quad ((u,v) \in E), \\
&\ \displaystyle\sum_{u \in C} x_u \le \floor{\frac{|C|}{2}} &\quad (C \in \mathcal{C}), \\
&\ x_u \in \{0,1\}^V,
\end{array}
\end{align}
where $\mathcal{C}$ is the set of odd cycles.
We assume that the graph is t-perfect. 
Then, Algorithm~\ref{alg:adaptive} has an approximation ratio of $(1 - \epsilon)$ with high probability for sufficiently large $T$,
and the number of revealed vertices is $O(\alpha T)$.

The only remaining issue is the number of iterations.
Since the system \eqref{eq:tperfect} is not required to be TDI\footnote{A graph is \emph{strongly t-perfect} if the system in \eqref{eq:tperfect} is TDI.
Any strongly t-perfect graph is t-perfect, but the converse is open.}, we have to discretize the dual variables.
We use Theorem~\ref{thm:kolliopoulos2005}. 
Let $\theta = \Theta\left(\frac{\epsilon^2}{\log n}\right)$.
The corresponding dual problem is given by
\begin{align}
\label{eq:tperfectdual}
\begin{array}{lll}
\text{minimize} &\ \displaystyle\sum_{e \in E} y_e + \sum_{C \in \mathcal{C}} \floor{\frac{|C|}{2}} \tilde z_C \\
\text{subject to} &\ \displaystyle\sum_{e \in \delta(u)} y_e + \sum_{C \in \mathcal{C}\colon u \in C} z_C \ge \tilde c_u &\quad (u \in V), \\
&\ y \in \mathbb{R}_+^E,\ z \in \mathbb{R}_+^\mathcal{C}.
\end{array}
\end{align}
We regard $z_C$ as a multiset in which each odd cycle $C$ contributes $\theta$.
Let $c_1, \ldots, c_k$ be the sizes of each odd cycles.
We then have $c_1 + \cdots + c_k < \mu / \theta$.
We define the witness cover by
\begin{align}
	W = \left\{\, (y, z) \in \theta \left(\mathbb{Z}_+^E \times \mathbb{Z}_+^{\mathcal{C}}\right) \, : \
        \sum_{e \in E} y_e + \sum_{C \in \mathcal{C}} \floor{\frac{|C|}{2}} z_C\le \left(1 - \frac{\epsilon}{2}\right) \mu \,\right\}.
\end{align}
\begin{Claim}
$|W| \le e^{O((\mu/\theta) \log n)}$.
\end{Claim}
\begin{proof}
To evaluate the size of $W$, we count $y$ and $z$ separately.
The number of candidates for $y$ is clearly $e^{O((\mu/\theta) \log m)}$ (and $m = O(n^2)$),
while the number of candidates for $z$ is bounded by $e^{O((\mu/\theta) \log n)}$ as the similar argument to Claim~\ref{claim:nonbipartite}.
%non-bipartite matching problem.
%\begin{align}
%	|W| = \sum_{c_1 + \cdots + c_k \le \mu/\theta} (n)_{c_1} \cdots (n)_{c_k} \le e^{O((\mu/\theta) \log n)}
%\end{align}
%where $(n)_{c} = n (n-1) \cdots (n-c+1)$. \COMM{YY}{The same as the non-bipartite matching?}
\end{proof}
Therefore, we obtain the following.
\begin{corollary}
By taking $T = \Omega(\Delta_c \log n \log (n/\epsilon)/\epsilon^3 p)$,
for the t-stable set problem, Algorithm~\ref{alg:adaptive} outputs a $(1 - \epsilon)$-approximate solution with probability at least $1 - \epsilon$.
\end{corollary}

\section{Vertex Sparsification Lemma}
\label{sec:sparsification}	

\subsection{Vertex Sparsification Lemma}

For the (unweighted) stochastic matching problem, Assadi et al.~\cite{assadi2016stochastic} proposed a procedure called \emph{vertex sparsification}, which reduces the number of vertices proportional to the maximum matching size $\mu$ while approximately preserving any matchings of size $\nu = \omega(1)$ with high probability.
This procedure is very useful as a preprocessing step for this problem since it makes $n/\mu = O(1)$, and so the required number of iterations becomes constant.

Here, we extend this procedure to an independence system on a $k$-uniform hypergraph and improve the result to preserve \emph{any} independence set with high probability without assuming $\nu = \omega(1)$.
In next section. we improve the performances of the algorithms for the bipartite matching problem, $k$-hypergraph matching problem, and $k$-column sparse packing integer programming problem by using this lemma.

In general, sparsification procedures are kinds of \emph{kernelization} procedure, which is studied in the area of parametrized complexity~\cite{downey2012parameterized}.
In particular, our and Assadi et al.~\cite{assadi2016stochastic}'s procedures are similar to the one in \cite{chitnis2016kernelization}, which aims to reduce space complexity of packing problems in streaming setting, but the conducted analyses and the provided guarantees are both different.

\medskip

Let $(V, E)$ be a $k$-uniform hypergraph and $(E, \mathcal{I})$ be an independence system
(which is a nonempty, downward-closed set system, i.e., ${\mathcal I} \neq \emptyset$, and $X \subseteq Y \in {\mathcal I} \implies X \in {\mathcal I}$),
whose rank function $r \colon 2^E \to {\mathbb Z}_+$ is defined by $r(S) = \max\{\, |I| : I \subseteq S,~I \in {\mathcal I} \,\}$.
We focus on the following special case of the stochastic packing integer programming problem \eqref{eq:packing} in this section:
\begin{align}\label{eq:sparsifiable}
\begin{array}{lll}
\text{maximize} &\ \displaystyle\sum_{e \in E} \tilde c_e x_e \\
\text{subject to} &\ \displaystyle\sum_{e \in S} x_e \le r(S) &\quad (S \subseteq E), \\
&\ x \in \{0,1\}^E.
\end{array}
\end{align}
Note that the constraint is equivalent to $\supp(x) \in {\mathcal I}$,
and this formulation still includes the $k$-column sparse PIP \eqref{eq:k-column_sparse}
(and hence all the matching problems shown in Section~\ref{sec:matching_problems}) as follows:
let $V = \{1, \ldots, n\}$ and $E = \{1, \ldots, m\}$
such that each hyperedge $j \in E$ is associated with a subset $\{\, i \in V : a_{ij} \neq 0 \,\}$
(if the size is less than $k$, add arbitrary vertices $i$ with $a_{ij} = 0$),
and define ${\mathcal I} = \{\, S \subseteq E : \sum_{j \in S} a_{ij} \leq b_i \ (\forall i \in V)\,\}$.

\begin{algorithm}[t]
\caption{Vertex sparsification.}
\label{alg:sparsification}
\begin{algorithmic}[1]
\State{Assign a random color in $\{1, \ldots, \frac{\beta(k,\epsilon,\delta) k^2 s}{2 \delta} \}$ to each vertex, where $\beta(k,\epsilon,\delta) = \frac{2 e^{\epsilon/k} \log (1 / \delta)}{\epsilon}$.}
\State{Return all colorful hyperedges.}
\end{algorithmic}
\end{algorithm}

Our procedure is shown in Algorithm~\ref{alg:sparsification}, which is a kind of color coding.
Let $s \in \mathbb{Z}_+$ be an upper bound on $r = r(E)$,
and $\epsilon, \delta \in (0, 1)$ be parameters for the accuracy and the probability, respectively.
It first assigns a random color in $\{1, \ldots, n^\circ\}$ to each vertex, where $n^\circ = \beta(k, \epsilon, \delta) k^2 s / \delta$ with $\beta(k, \epsilon, \delta) = 2 e^{\epsilon/k} \log (1/\delta)/\epsilon$.
It then returns all ``colorful'' hyperedges that consists entirely of differently colored vertices.
This yields an independence system on the color class consisting from $n^\circ = \Theta(s)$ vertices.

\begin{lemma}[Vertex Sparsification Lemma]
\label{lem:sparsification}
Suppose that $n \ge 2 k$. 
Then, after Algorithm~\ref{alg:sparsification},
for any independent set $I \in \mathcal{I}$ in the original instance,
there exists an independent set $I^\circ \subseteq I$ of size at least $(1 - \epsilon) |I|$
in the sparsified instance with probability at least $1 - \delta$.
\end{lemma}
\begin{proof} 
Let $\nu = |I|$. 
For notational simplicity, we denote by $\beta = \beta(k, \epsilon, \delta)$.
We now make the following case analysis.

\paragraph{Case 1: $\nu \le \beta$ (the rank of $I$ is small).}

If all vertices incident to $I$ have different colors, the size of $I$ is preserved after the mapping. 
Since the number of the incident vertices is at most $k \nu$, the probability that this has occurred is at least
\begin{align}
	\frac{n^\circ (n^\circ - 1) \cdots (n^\circ - k \nu + 1)}{n^{\circ k \nu}} 
    \ge \exp \left( - \frac{k^2 \nu^2}{n^\circ} \right) 
   \ge \exp \left( - \frac{\delta \nu^2}{\beta s} \right) \ge e^{-\delta} \ge 1 - \delta.
\end{align}
Here, the first inequality follows from the falling factorial approximation (the next lemma), and the second inequality follows from $\nu \le r \le s$ and $\nu \le \beta$.
\begin{lemma}[Falling Factorial Approximation]
\label{thm:fallingfactorial}
\begin{align}
	\frac{n (n-1) \cdots (n-k+1)}{n^k} \ge \exp \left( -\frac{k^2}{n} \right).
\end{align}
\end{lemma}
\begin{proof}
Recall that $\log (1 - x) \ge - x/(1-x)$ for all $x \in (0,1)$. 
The logarithm of the above is
\begin{align}
	\quad &\sum_{i=1}^{k-1} \log \left(1 - \frac{i}{n}\right) \ge -\sum_{i=1}^{k-1} \frac{i}{n - i} \ge -\sum_{i=1}^{k-1} \frac{i}{n - k} 
    = -\frac{k (k-1)}{2(n - k)} \ge -\frac{k^2}{n}. \quad \qedhere
\end{align}
%\COMM{YY}{The last inequality does not hold...}
\end{proof}
%Here, the first inequality is the falling factorial approximation (Theorem~\ref{thm:fallingfactorial}), and the second inequality is given by $\nu \le s$ and $\nu \le \beta$.

\paragraph{Case 2: $\nu \ge \beta$ (the rank of $I$ is large).}

We further reduce the number of colors by mapping each color class to $\{1, \ldots, k^2 \nu/\epsilon \}$. (Note that $k^2 \nu/\epsilon \le n^\circ$ since $\beta \ge 1/\epsilon$.)
We say that a color class $c$ is \emph{good} if some vertex in color $c$ is covered by some hyperedge $e \in I$, and otherwise we say that $c$ is \emph{bad}.

For each color class $c$,
let $X_c$ be the indicator of the event that $c$ is bad,
i.e., $X_c = 1$ if $c$ is bad and $X_c = 0$ otherwise.
Then $\Pr(X_c = 1) = (1 - \epsilon/k^2 \nu)^{k \nu} \le e^{-\epsilon/k}$. 
Therefore $\E\left[\sum_c X_c\right] \le e^{-\epsilon/k} k^2 \nu / \epsilon$. 
Since $X_c$ are negatively correlated random variables, we can apply the Chernoff bound~\cite{panconesi1997randomized}:
\begin{align}
  \Pr\left(\sum_c X_c \ge \frac{k^2 \nu}{\epsilon} - \left(1 - \frac{\epsilon}{k}\right) k \nu\right)
  &= \Pr\left(\sum_c X_c \ge \left(1 - \frac{\epsilon}{k} + \frac{\epsilon^2}{k^2}\right) \frac{k^2 \nu}{\epsilon} \right) \notag \\
  &\le \Pr\left(\sum_c X_c \ge \left(1 + \frac{\epsilon^2}{2 k^2}\right) e^{-\epsilon/k} \frac{k^2 \nu}{\epsilon} \right) \notag \\
  &\le \exp\left(-\epsilon e^{-\epsilon/k} \frac{\nu}{2}\right) \le \exp\left(-\epsilon e^{-\epsilon/k} \frac{\beta}{2}\right) = \delta, 
\end{align}
where the first inequality follows from $(1+x^2/2) e^{-x} \le 1 - x + x^2$ and the last equality follows from the definition of $\beta$.
Therefore, there are at least $(1 - \epsilon/k) k\nu$ good color classes with high probability.

For each good color class, we select one covered vertex and remove all other vertices.
The number of removed vertices is at most $\epsilon \nu$, so at most $\epsilon \nu$ hyperedges in the independent set are removed.
The remaining hyperedges form an independent set of size at least $(1 - \epsilon) \nu$.
\end{proof}

\begin{remark}
The second part is a simple extension of Assadi et al.~\cite{assadi2016stochastic}.
Since they only analyzed this case, $\nu = \omega(1)$ was required.
\end{remark}

\subsection{Usage of Vertex Sparsification Lemma}

Here, we describe how to use the vertex sparsification lemma to improve the performance of Algorithms~\ref{alg:adaptive} and \ref{alg:nonadaptive}.
For simplicity, we only describe the result for Algorithm~\ref{alg:adaptive}, as Algorithm~\ref{alg:nonadaptive} can be handled using the same argument.

Let $(V, E)$ be a $k$-uniform hypergraph with $|V| = n$ and $|E| = m$ and $(E, \mathcal{I})$ be an independence system.
We consider the problem \eqref{eq:sparsifiable}, where we assume the following.
\begin{enumerate}
\item There exists an LP-relative $\alpha$-approximation algorithm.
\item The number of iterations required to guarantee $(1 - \epsilon) \alpha$-approximation with probability at least $1 - \delta$ is bounded by $T(\log (n/\mu), \epsilon, \delta)$.
\end{enumerate}

The method is shown in Algorithm~\ref{alg:speedup}, where $\epsilon, \delta \in (0, 1)$ are parameters for the accuracy and the probability, respectively, and $c_{\max} = \max_j c_j^+$.
We first estimate the maximum size $s$ of the independent sets such that $\alpha s \le r \le s$, which is computed via LP relaxation.
We then apply Algorithm~\ref{alg:sparsification} to obtain a sparsified instance, and finally apply Algorithm~\ref{alg:adaptive} or \ref{alg:nonadaptive} with an LP-relative $\alpha$-approximation algorithm to obtain a solution.
We now analyze the performance of this procedure.

\begin{algorithm}[tb]
\caption{Speedup by vertex sparsification.}
\label{alg:speedup}
\begin{algorithmic}[1]
\State{Estimate the size $s$ of maximum independent set such that $\alpha s \le r \le s$.}
\State{Sparsify the instance by Algorithm~\ref{alg:sparsification} with accuracy parameter $\epsilon' = \epsilon/(1 + c_\text{max})$ and probability parameter $\delta' = \delta/4$.}
\State{Run Algorithm~\ref{alg:adaptive} or \ref{alg:nonadaptive} with an LP-relative $\alpha$-approximation algorithm
by setting $T = T(O(\log (k / p \alpha \epsilon' \delta')), \epsilon', \delta')$.}\label{line:3''}
\end{algorithmic}
\end{algorithm}

\begin{theorem}
%Suppose that $\epsilon < 1/2(1 + c_{\max})$, where $c_{\max} = \max_j c_j^+$.
Algorithm~\ref{alg:speedup} finds a $(1 - \epsilon) \alpha$-approximate solution with probability at least $1 - \delta$.
\end{theorem}

\begin{proof}
Let $r = r(E)$ be the rank of the original independence system and $\tilde\mu$ the optimal value of the original instance \eqref{eq:sparsifiable}, which is a random variable determined by nature. 
Also, let $r^\circ$ be the rank of the sparsified instance, and $\tilde \mu^\circ$ be the optimal value of the sparsified instance, which is also a random variable.

\begin{Claim}\label{cl:sparsified_bound}%\mbox{}\vspace{-5mm}
\begin{align}
  \Pr\left(r^\circ \ge (1 - \epsilon')r\right) &\geq 1 - \delta',\label{eq:sparsified_bound_1}\\
  \Pr\left(\tilde \mu^\circ \ge (1 - c_{\max} \epsilon') \tilde \mu \right) &\geq 1 - \delta'.\label{eq:sparsified_bound_2}
\end{align}
\end{Claim}

\begin{proof}
The first inequality \eqref{eq:sparsified_bound_1} immediately follows from Lemma~\ref{lem:sparsification},
and we focus on the second \eqref{eq:sparsified_bound_2}.
Fix a realization of $\tilde{c}$,
and let $x \in \{0, 1\}^E$ be an optimal solution to \eqref{eq:sparsifiable}
such that $I = \supp(x) \in {\mathcal I}$ is minimal.
By Lemma~\ref{lem:sparsification},
there exists an independent set $I^\circ \subseteq I$ of size $|I^\circ| \geq (1 - \epsilon')|I|$
in the sparsified instance with probability $1 - \delta'$.
Let $x^\circ \in \{0, 1\}^E$ be the vector with $\supp(x^\circ) = I^\circ$,
whose restriction to the sparsified hyperedge set is a feasible solution to the sparsified instance.
We then have
\begin{align}
  \tilde\mu^\circ \geq \tilde c^\top x^\circ \geq \tilde c^\top x -  c_{\max} \epsilon'|I| \geq \tilde\mu - c_{\max} \epsilon' \tilde\mu,
\end{align}
where the last inequality follows from the minimality of $I = \supp(x)$
(for each $j \in I$, we must have $\tilde{c}_j \geq 1$, and hence $\tilde\mu = \tilde c^\top x \geq \supp(x) = |I|$).
\end{proof}

By Claim \ref{cl:sparsified_bound},
we have $r^\circ \ge (1 - \epsilon')r$ and $\tilde \mu^\circ \ge (1 - c_{\max} \epsilon') \tilde \mu$
with probability at least $1 - 2\delta'$.
Under this event, by using Algorithm~\ref{alg:adaptive} in Line~\ref{line:3''} of Algorithm~\ref{alg:speedup},
we obtain a solution whose objective value is at least $(1 - \epsilon') \tilde \mu^\circ \ge (1 - (c_{\max} + 1) \epsilon') \tilde \mu = (1 - \epsilon) \tilde \mu$ with probability at least $1 - \delta'$
(and hence with probability at least $1 - 3\delta'$ in total).

The remaining issue is the number of iterations.
That is, for $\beta' = \beta(k, \epsilon', \delta') = 2e^{\epsilon'/k}\log(1/\delta')/\epsilon'$ and $n^\circ = \beta'k^2s/\delta'$, we prove
\begin{align}\label{eq:sparsified_ratio}
  \log\frac{n^\circ}{\tilde\mu^\circ} = O\left(\log \frac{k}{p \alpha \epsilon' \delta'}\right),
\end{align}
with probability at least $1 - \delta'$,
which implies that we succeed with probability at least $1 - 4\delta' = 1 - \delta$ through Algorithm \ref{alg:speedup}.
We make a case analysis.

\paragraph{Case 1. $r^\circ \ge 8\log(1 / \delta')/p$ (the rank of the sparsified instance is large).}

We evaluate the objective value of the independent set in the sparsified instance that corresponds to the maximum independent set in the original instance.
Since each element in the sparsified independent set contributes at least $1$ with probability at least $p$, we can apply the Chernoff bound
\begin{align}
	\Pr\left( \tilde \mu^\circ \ge \frac{p r^\circ}{2} \right) \ge \Pr\left( \sum_{i=1}^{r^\circ} X_i \ge \frac{p r^\circ}{2} \right)
    \ge 1 - e^{-p r^\circ / 8} \ge 1 - \delta',
\end{align}
where $X_i$ ($i = 1, \ldots, r^\circ$) are i.i.d.\ random variables following the Bernoulli distribution with probability $p$.
Under this event ($\tilde \mu^\circ \ge p r^\circ / 2$), we have 
\begin{align}
	\frac{n^\circ}{\tilde \mu^\circ} \le \frac{2 n^\circ}{p r^\circ} \le \frac{4 n^\circ}{p r} \le \frac{4 n^\circ}{p \alpha s} = \frac{4 \beta' k^2}{p \alpha \delta'} = \frac{8e^{\epsilon'/k}\log(1/\delta')}{p\alpha\epsilon'\delta'},
\end{align}
where the second inequality follows from $r^\circ \geq (1 - \epsilon')r \geq r/2$ (because $\epsilon' = \epsilon/(1 + c_{\max}) \leq 1/2$), and
the third from $r \geq \alpha s$.
Since $e^{\epsilon'/k} = O(1)$ and $\log(1/\delta') \leq 1/\delta'$, we have \eqref{eq:sparsified_ratio}.

\paragraph{Case 2. $r^\circ \le 8\log(1 / \delta')/p$ (the rank of the sparsified instance is small).}
We have
\begin{align}
	\frac{n^\circ}{\tilde \mu^\circ} \le n^\circ \leq \frac{\beta' k^2 r}{\alpha \delta'} \le \frac{2 \beta' k^2 r^\circ}{\alpha \delta'} \le \frac{16 \beta' k^2 \log(1/\delta')}{p \alpha \delta'},
\end{align}
where the first inequality follows from $\tilde\mu^\circ \geq 1$ (because it is an integer with $\tilde\mu^\circ \geq (1 - \epsilon' c_{\max})\tilde\mu > 0$),
the second from $r \geq \alpha s$, and the third from $r^\circ \geq r/2$.
This leads to \eqref{eq:sparsified_ratio} similarly in Case 1.
\end{proof}
%\begin{Remark}
%Assadi et al.~\cite{assadi2016stochastic} only analyzed Case~1, so they assumed that $\mu^* = \omega(1/p)$.
%Here, we have successfully removed this assumption.
%\end{Remark}

The sizes of witness covers of bipartite matching, $k$-hypergraph matching, and $k$-column-sparse packing integer programming depend on $n/\mu$. 
Thus these are improved by using this technique.

\begin{corollary}\label{cor:bipartite_stronger}
For the bipartite matching problem with $c_j = O(1)$ for all $j$, there is an algorithm that conducts $O(\log(1/\epsilon p) / \epsilon p)$ queries per vertex and finds $(1 - \epsilon)$-approximate solution with probability at least $1 - \epsilon$.
\end{corollary}

\begin{corollary}
For the $k$-hypergraph matching problem with $c_j = O(1)$ for all $j$, there is an algorithm that conducts $O(k (\log(k /\epsilon p) + 1/\epsilon)/\epsilon p)$ queries per vertex and finds $(1 - \epsilon)/(k - 1 + 1/k)$-approximate solution with probability at least $1 - \epsilon$.
\end{corollary}

\begin{corollary}
For the $k$-column sparse packing integer programming problem with $c_j = O(1)$, $b_i = O(1)$, and $A_{ij} = O(1)$ for all $i, j$, there is an algorithm that conducts $O(k (\log(k /\epsilon p) + 1/\epsilon)/\epsilon p)$ queries per vertex and finds $(1 - \epsilon)/2 k$-approximate solution with probability at least $1 - \epsilon$.
\end{corollary}

\section*{Acknowledgments}
The authors thank anonymous reviewers for their careful reading and a number of valuable comments.
This work was supported by JSPS KAKENHI Grant Numbers 16H06931 and 16K16011.

\bibliographystyle{plain}
\bibliography{main}

\begin{thebibliography}{10}

\bibitem{adamczyk2011improved}
Marek Adamczyk.
\newblock Improved analysis of the greedy algorithm for stochastic matching.
\newblock {\em Information Processing Letters}, 111(15):731--737, 2011.

\bibitem{adamczyk2016submodular}
Marek Adamczyk, Maxim Sviridenko, and Justin Ward.
\newblock Submodular stochastic probing on matroids.
\newblock {\em Mathematics of Operations Research}, 41(3):1022--1038, 2016.

\bibitem{assadi2016stochastic}
Sepehr Assadi, Sanjeev Khanna, and Yang Li.
\newblock The stochastic matching problem with (very) few queries.
\newblock In {\em Proceedings of the 17th ACM Conference on Economics and
  Computation}, pages 43--60. ACM, 2016.

\bibitem{assadi2017stochastic}
Sepehr Assadi, Sanjeev Khanna, and Yang Li.
\newblock The stochastic matching problem: beating half with a non-adaptive
  algorithm.
\newblock In {\em Proceedings of the 18th ACM Conference on Economics and
  Computation}, pages 99--116. ACM, 2017.

\bibitem{bansal2012lp}
Nikhil Bansal, Anupam Gupta, Jian Li, Juli{\'a}n Mestre, Viswanath Nagarajan,
  and Atri Rudra.
\newblock When {LP} is the cure for your matching woes: improved bounds for
  stochastic matchings.
\newblock {\em Algorithmica}, 63(4):733--762, 2012.

\bibitem{behnezhad2017almost}
Soheil Behnezhad and Nima Reyhani.
\newblock Almost optimal stochastic weighted matching with few queries.
\newblock In {\em Proceedings of the 16th ACM Conference on Economics and
  Computation}, pages 235--249. ACM, 2018.

\bibitem{blum2015ignorance}
Avrim Blum, John~P. Dickerson, Nika Haghtalab, Ariel~D. Procaccia, Tuomas
  Sandholm, and Ankit Sharma.
\newblock Ignorance is almost bliss: near-optimal stochastic matching with few
  queries.
\newblock In {\em Proceedings of the 16th ACM Conference on Economics and
  Computation}, pages 325--342. ACM, 2015.

\bibitem{blum2013harnessing}
Avrim Blum, Anupam Gupta, Ariel Procaccia, and Ankit Sharma.
\newblock Harnessing the power of two crossmatches.
\newblock In {\em Proceedings of the 14th ACM conference on Electronic
  Commerce}, pages 123--140. ACM, 2013.

\bibitem{chan2010linear}
Yuk~Hei Chan and Lap~Chi Lau.
\newblock On linear and semidefinite programming relaxations for hypergraph
  matching.
\newblock In {\em Proceedings of the 21st Annual ACM-SIAM Symposium on Discrete
  Algorithms}, pages 1500--1511. SIAM, 2010.

\bibitem{chen2009approximating}
Ning Chen, Nicole Immorlica, Anna~R. Karlin, Mohammad Mahdian, and Atri Rudra.
\newblock Approximating matches made in heaven.
\newblock In {\em Proceedings of the 36th International Colloquium on Automata,
  Languages, and Programming}, pages 266--278. Springer, 2009.

\bibitem{chitnis2016kernelization}
Rajesh Chitnis, Graham Cormode, Hossein Esfandiari, MohammadTaghi Hajiaghayi,
  Andrew McGregor, Morteza Monemizadeh, and Sofya Vorotnikova.
\newblock Kernelization via sampling with applications to finding matchings and
  related problems in dynamic graph streams.
\newblock In {\em Proceedings of the 27th Annual ACM-SIAM Symposium on Discrete
  Algorithms}, pages 1326--1344. SIAM, 2016.

\bibitem{costello2012stochastic}
Kevin Costello, Prasad Tetali, and Pushkar Tripathi.
\newblock Stochastic matching with commitment.
\newblock In {\em Proceedings of the 39th International Colloquium on Automata,
  Languages and Programming}, pages 822--833. Springer, 2012.

\bibitem{cunningham1978primal}
William~H. Cunningham and A.~B. Marsh.
\newblock A primal algorithm for optimum matching.
\newblock In {\em Polyhedral Combinatorics}, pages 50--72. Springer, 1978.

\bibitem{dean2004approximating}
Brian~C. Dean, Michel~X. Goemans, and Jan Vondr{\'a}k.
\newblock Approximating the stochastic knapsack problem: the benefit of
  adaptivity.
\newblock In {\em Proceedings of the 45th Annual IEEE Symposium on Foundations
  of Computer Science}, pages 208--217. IEEE, 2004.

\bibitem{dean2005adaptivity}
Brian~C. Dean, Michel~X. Goemans, and Jan Vondr{\'a}k.
\newblock Adaptivity and approximation for stochastic packing problems.
\newblock In {\em Proceedings of the 16th Annual ACM-SIAM Symposium on Discrete
  Algorithms}, pages 395--404. SIAM, 2005.

\bibitem{dickerson2016organ}
John~P. Dickerson and Tuomas Sandholm.
\newblock Organ exchanges: a success story of {AI} in healthcare.
\newblock In {\em Thirtieth Conference on Artificial Intelligence Tutorial
  Forum}, 2016.

\bibitem{downey2012parameterized}
Rodney~G. Downey and Michael~Ralph Fellows.
\newblock {\em Parameterized Complexity}.
\newblock Springer Science \& Business Media, 2012.

\bibitem{edmonds1970submodular}
Jack Edmonds.
\newblock Submodular functions, matroids, and certain polyhedra.
\newblock In {\em Proceedings of the Calgary International Conference on
  Combinatorial Structures and Their Applications}, pages 69--87. Gordon and
  Breach, 1970.

\bibitem{egervary1931combinatorial}
Eugene Egerv{\'a}ry.
\newblock On combinatorial properties of matrices.
\newblock {\em Matematikai \'{e}s Fizikai Lapok}, 38:16--28, 1931.

\bibitem{gijswijt2013algorithm}
Dion Gijswijt and Gyula Pap.
\newblock An algorithm for weighted fractional matroid matching.
\newblock {\em Journal of Combinatorial Theory, Series B}, 103(4):509--520,
  2013.

\bibitem{gupta2013stochastic}
Anupam Gupta and Viswanath Nagarajan.
\newblock A stochastic probing problem with applications.
\newblock In {\em Proceedings of the 16th International Conference on Integer
  Programming and Combinatorial Optimization}, pages 205--216. Springer, 2013.

\bibitem{gupta2017adaptivity}
Anupam Gupta, Viswanath Nagarajan, and Sahil Singla.
\newblock Adaptivity gaps for stochastic probing: submodular and {XOS}
  functions.
\newblock In {\em Proceedings of the 28th Annual ACM-SIAM Symposium on Discrete
  Algorithms}, pages 1688--1702. SIAM, 2017.

\bibitem{hurkens1989size}
Cor A.~J. Hurkens and Alexander Schrijver.
\newblock On the size of systems of sets every $t$ of which have an {SDR}, with
  an application to the worst-case ratio of heuristics for packing problems.
\newblock {\em SIAM Journal on Discrete Mathematics}, 2(1):68--72, 1989.

\bibitem{kiraly2012degree}
Tam{\'a}s Kir{\'a}ly, Lap~Chi Lau, and Mohit Singh.
\newblock Degree bounded matroids and submodular flows.
\newblock {\em Combinatorica}, 32(6):703--720, 2012.

\bibitem{kolliopoulos2005approximation}
Stavros~G. Kolliopoulos and Neal~E. Young.
\newblock Approximation algorithms for covering/packing integer programs.
\newblock {\em Journal of Computer and System Sciences}, 71(4):495--505, 2005.

\bibitem{kHonig1931graphs}
D\'{e}nes K{\H{o}}nig.
\newblock Graphs and matrices.
\newblock {\em Matematikai \'{e}s Fizikai Lapok}, 38:116--119, 1931.

\bibitem{lee2013matroid}
Jon Lee, Maxim Sviridenko, and Jan Vondr{\'a}k.
\newblock Matroid matching: the power of local search.
\newblock {\em SIAM Journal on Computing}, 42(1):357--379, 2013.

\bibitem{molinaro2011query}
Marco Molinaro and R.~Ravi.
\newblock The query-commit problem.
\newblock {\em arXiv preprint arXiv:1110.0990}, 2011.

\bibitem{panconesi1997randomized}
Alessandro Panconesi and Aravind Srinivasan.
\newblock Randomized distributed edge coloring via an extension of the
  {C}hernoff--{H}oeffding bounds.
\newblock {\em SIAM Journal on Computing}, 26(2):350--368, 1997.

\bibitem{parekh2011iterative}
Ojas Parekh.
\newblock Iterative packing for demand and hypergraph matching.
\newblock In {\em Proceedings of the 15th International Conference on Integer
  Programming and Combinatorial Optimization}, pages 349--361. Springer, 2011.

\bibitem{parekh2014generalized}
Ojas Parekh and David Pritchard.
\newblock Generalized hypergraph matching via iterated packing and local ratio.
\newblock In {\em Proceedings of the 12th International Workshop on
  Approximation and Online Algorithms}, pages 207--223. Springer, 2014.

\bibitem{raghavan1987randomized}
Prabhakar Raghavan and Clark~D. Tompson.
\newblock Randomized rounding: a technique for provably good algorithms and
  algorithmic proofs.
\newblock {\em Combinatorica}, 7(4):365--374, 1987.

\bibitem{roth2004kidney}
Alvin~E. Roth, Tayfun S{\"o}nmez, and M.~Utku {\"U}nver.
\newblock Kidney exchange.
\newblock {\em The Quarterly Journal of Economics}, 119(2):457--488, 2004.

\bibitem{schrijver2003combinatorial}
Alexander Schrijver.
\newblock {\em Combinatorial Optimization: Polyhedra and Efficiency}.
\newblock Springer Science \& Business Media, 2003.

\bibitem{shepherd2007demand}
F.~Bruce Shepherd and Adrian Vetta.
\newblock The demand-matching problem.
\newblock {\em Mathematics of Operations Research}, 32(3):563--578, 2007.

\bibitem{turan1954theory}
Paul Tur{\'a}n.
\newblock On the theory of graphs.
\newblock In {\em Colloquium Mathematicum}, pages 19--30, 1954.

\end{thebibliography}

\end{document}